\documentclass[10pt, conference]{IEEEtran}
\IEEEoverridecommandlockouts

\usepackage{amsmath}
\usepackage{amsthm}
\usepackage{amsfonts}
\usepackage{amssymb}
\usepackage{bm}
\usepackage{microtype}
\usepackage{mathrsfs}
\usepackage{mathtools}
\usepackage{graphicx}
\usepackage{cite}
\usepackage{epstopdf}
\usepackage{bbm}
\newtheorem{theorem}{Theorem}

\newtheorem{remark}{Remark}
\usepackage{algorithm}
\usepackage{algpseudocode}

\algrenewcommand\alglinenumber[1]{\scriptsize #1}

\makeatletter
\let\OldStatex\Statex
\renewcommand{\Statex}[1][3]{%
  \setlength\@tempdima{\algorithmicindent}%
  \OldStatex\hskip\dimexpr#1\@tempdima\relax}
\makeatother
\algdef{SE}[DOWHILE]{Do}{DoWhile}{\algorithmicdo}[1]{\algorithmicwhile\ #1}

\usepackage{amsfonts}
\usepackage{times}
\usepackage{latexsym}
\usepackage{amssymb}
\usepackage{amsmath}
\usepackage{cite}
\usepackage{verbatim}
\newtheorem{thm}{Theorem}
\newtheorem{cor}[theorem]{Corollary}

\newtheorem{defn}[theorem]{Definition}

\newtheorem{proposition}[theorem]{Proposition}

\def\SNR{{\textsf{SNR}}}

\def\SINR{{\mathsf{SINR}}}

\def\bb0{{\mathbb{0}}}

\def\bb{{\mathbf{b}}}

\def\b0{{\mathbf{0}}}

\def\bB{{\mathbf{B}}}

\def\b1{{\mathbf{1}}}

\def\bbE{{\mathbb{E}}}

\def\bbP{{\mathbb{P}}}

\def\bbR{{\mathbb{R}}}

\def\cF{\mathcal{F}}
\def\cG{\mathcal{G}}

\def\cI{\mathcal{I}}

\def\cS{\mathcal{S}}

\def\cU{\mathcal{U}}

\def\sfN{\mathsf{N}}

\def\sfR{\mathsf{R}}

\def\sfm{{\mathsf{m}}}
\def\sfn{{\mathsf{n}}}

\def\sfu{{\mathsf{u}}}

\def\sf0{{\mathsf{0}}}

\def\nn{\nonumber}

\newcommand{\ep}{\epsilon}
\newcommand{\al}{\alpha}

\def\ga{\gamma}

\def\1{\mathbf{1}}

\def\f{\frac}

\title{Capacity of Cellular Wireless Network}
\author{\IEEEauthorblockN{Rahul Vaze}
\IEEEauthorblockA{School of Technology and Computer Science \\ TIFR Mumbai, India \\
vaze@tcs.tifr.res.in}\thanks{The research was supported by Young Scientist Award grant from the Indian National Science Academy to Rahul Vaze.}\thanks{This paper in part will appear in Proc. WiOpt 2017, May 15-19, 2017, Paris.}
\and
\IEEEauthorblockN{Srikanth K. Iyer}
\IEEEauthorblockA{Department of Mathematics\\Indian Institute of Science, Bangalore,\\
skiyer@math.iisc.ernet.in}

}

\begin{document}
\maketitle
\thispagestyle{empty}
\pagestyle{empty}
\begin{abstract}
Earlier definitions of capacity for wireless networks, e.g., transport or transmission capacity, for which exact theoretical results are known, are well suited for ad hoc networks but are not directly applicable for cellular wireless networks, where large-scale basestation (BS) coordination is not possible, and retransmissions/ARQ under the SINR model is a universal feature.

In this paper, cellular wireless networks, where both BS locations and mobile user (MU) locations are distributed as independent Poisson point processes are considered, and each MU connects to its nearest BS.
With ARQ, under the SINR model, the effective downlink rate of packet transmission is the reciprocal of the expected delay (number of retransmissions needed till success), which we use as our network capacity definition after scaling it with the BS density.

Exact characterization of this natural capacity metric for cellular wireless networks is derived.
The capacity is shown to first increase polynomially with the BS density in the low BS density regime and then scale inverse exponentially with the increasing BS density. Two distinct upper bounds are derived that are relevant for the low and the high BS density regimes.  A single power control strategy is shown to achieve the upper bounds in both the regimes.
This result is fundamentally different from the well known capacity results for ad hoc networks, such as transport and transmission capacity that scale as the square root of the (high) BS density. Our results show that the strong temporal correlations of SINRs with PPP distributed BS locations is limiting, and the 
realizable capacity in cellular wireless networks in high-BS density regime is much smaller than previously thought. 
A byproduct of our analysis shows that the capacity of the ALOHA strategy with retransmissions is zero. \end{abstract}


\section{Introduction}
Finding the Shannon capacity of a wireless network is perhaps one of the most well-studied problem that has remained unsolved. 
For ad hoc wireless networks, two slightly relaxed capacity notions have been defined, transport \cite{Gupta2000} and transmission \cite{Weber2005}, for which  theoretical results have been possible mainly because of two important simplifications; SINR model of communication, and assuming random locations for nodes \cite{bookvaze2015}.

Under the SINR model, communication between two nodes is deemed successful if the SINR between them is larger than a threshold that depends on the rate of transmission. The assumption of random location of nodes takes two forms, either nodes are assumed to be distributed uniformly on a unit-radius disc or nodes are assumed to be located on the whole of $\bbR^2$ as a homogenous 
Poisson point process (PPP) with a fixed density. 

Even under these simplifications, as far as we know, there has been no fundamental characterization of the maximum throughput possible (under any reasonable capacity definition) in {\it cellular} wireless networks that are structured, rather than being ad hoc. 
Two important features of cellular wireless networks that may have been behind the lack of capacity results are: preclusion of large scale basestation (BS) coordination/scheduling (allowed in transport capacity) and universal implementation of automatic repeat request (ARQ) protocol, where packets are transmitted repeatedly until successful reception via acks/nacks. Retransmissions entail incorporating temporal correlations of SINRs, thereby making the analysis challenging.

In this paper, we consider the well-accepted model of a cellular wireless network, also called the tractable model \cite{andrews2011tractable}, that has BS locations process distributed as a homogenous PPP, and the main of focus is on characterizing the 'capacity' of the cellular network in the downlink. The tractable model is a reasonable abstraction in the modern scenario, where multiple layers of BSs (macro, micro, femto) are overlaid over each other. Moreover, we consider the widely used BS-MU (mobile user) association rule, where each MU connects to its nearest BS, i.e., all MUs lying in a Voronoi cell connect to the representative Voronoi BS. The MU locations are also assumed to be distributed as a homogenous PPP, independent of the BS locations' process.

We consider the SINR model of transmission for each BS-MU communication, where with ARQ, a packet is retransmitted from the BS until the SINR seen at the MU is above a certain threshold. Also, each BS serves all the MUs located in its Voronoi cell in a round-robin manner by dividing its slots/bandwidth equally among them to closely model the 'fair' practical implementation.

BSs in real-life cellular wireless networks are limited in their ability to coordinate their transmissions in order to control inter-cell interference. We begin by considering that each BS is only allowed to use local strategies, i.e., each BS's transmission decisions can only be based on local channel conditions (path-loss or fading gain) at the MU or feedback (ack/nacks) from the MU it is serving. Note that with ARQ, the local strategies are allowed to be adaptive, for example, a BS can use power control, or choose whether to transmit at all, given the history of ack/nacks. 
We later extend our model to include the realistic small-scale BS coordination, where nearby basestations can schedule their transmissions together, similar to coordinated multi-point (CoMP).

With ARQ, let $D$ be the number of retransmissions needed for a packet from BS $x$ to be successful at its MU $y$, defined as 
$$D = \min\{t: \SINR_{xy}(t) \ge \beta\}.$$ $D$ has the interpretation of delay, the time (or retransmissions) needed to receive the packet successfully. 

Let $\lambda$ be the BS density (per $m^2$) of the cellular network. With ARQ, to count for the actual rate of successful packet transmission, we consider a natural definition of {\bf capacity} (in the downlink) that is proportional toð the reciprocal of the expected delay. For analysis, we consider a typical MU, which without loss of generality is assumed to be located at the origin, that connects to its nearest BS. 
Let $n_0$ be the number of MUs in the Voronoi cell that contains the origin where the typical-MU is located. For any local strategy 
$\cS$ used by BSs, the per-MU capacity is 
$$C_m(\cS) = \frac{1}{n_0\bbE\{D\}} \ \text{packets/sec},$$
the per-BS capacity with $\cS$ under round-robin policy is
 \begin{equation*}
C_b(\cS) = \frac{1}{\bbE\{D\}}\ \text{packets/sec} .
\end{equation*} 
 and the
network wide capacity with $\cS$ is
 \begin{equation*}
C(\cS) = \frac{\lambda}{\bbE\{D\}}\ \text{packets/sec/m}^2 .
\end{equation*}

Hence, our capacity definition is 
 \begin{equation}
\label{eq:Cdefn}
C = \sup_\cS C(\cS),
\end{equation} i.e., we are looking for the best possible local (adaptive) BS strategy $\cS$ that achieves the maximum 
network-wide throughput. Since we are assuming that each BS serves all the MUs connected to itself in a round-robin manner, $C$ is independent of the MU density. Note that even though we are assuming round robin scheduling, the performance analysis is for a typical user (e.g. randomly chosen), and not the worst case user, e.g., the cell-edge users. Thus, the round robin policy is not a limiting factor.

The SINR model of transmission for the BS-MU link can be thought of as multi-ary-non-symmetic erasure channel, that is an extension of a binary-erasure-channel (BEC) with feedback, where each BS has multiple choices of transmit powers, and where the probability of erasure depends on the choice of the transmitted power. In comparison, in a BEC, there are only two possible transmission choices and the erasure probabilities (erasure probabilities) for both the choices are identical. As one may recall, the Shannon capacity for the BEC $C_{BEC}$ is achieved by a simple strategy of retransmitting the packet until it is correctly received, and  $C_{BEC}=\frac{1}{\bbE\{T\}}$ \cite{Cover2004}, where $T$ is total number of retransmissions needed. 
Our capacity definition \eqref{eq:Cdefn} is in similar spirit, where in addition, BS strategy can choose what power level to use after every erasure. One important distinction, however, compared to the usual BEC, is that the SINRs are temporally correlated, which in information theoretic language translates to the channel having memory.

\subsection{Prior work on Network Capacity}
Two related metrics of capacity; transport and transmission, have been defined for ad hoc networks and for both, 
exact results have been obtained. 
The transport capacity framework \cite{Gupta2000} assumes that $\lambda$ nodes are distributed uniformly on a unit-radius disc, and $\lambda/2$ source-destination pairs among them are chosen randomly. The transport capacity is the measure of how many bit-meters can be simultaneously transported across the network, where one bit-meter is transported if one
bit is successfully (SINR model) transmitted to a distance of one meter towards its destination. 
Under the path-loss only model, where fading gain is neglected, transport capacity has been shown to be $\Theta(\sqrt{\lambda})$ \cite{Gupta2000, Franceschetti2004}. Some extensions of transport capacity are also known, where with fast mobile node mobility it is known to scale as $\Theta(\lambda)$, and also in the information theoretic communication model, it scales as $\Theta(\lambda)$ \cite{TseScaling2007}.

Transmission capacity framework \cite{Weber2005} assumes the other type of randomness, where (transmitter) node locations are distributed as a homogenous PPP with density $\lambda$. Each transmitter has a receiver at a 
fixed distance $d_f$ from it in a random direction. 
For each $\ep$ (reliability constraint), the transmission capacity is defined as the largest density $\lambda$, such that the outage probability $\bbP(\SINR <\beta)$ for each transmission is less than $\ep$. 

Let $\lambda_\ep$ be the transmission capacity with reliability constraint of $1-\ep$. The distance scaled transmission capacity, where $\lambda_\ep$ is multiplied with the fixed transmitter-receiver distance $d_f$, is comparable with transport capacity,  and remarkably, also scales as $\Theta(\sqrt{\lambda_\ep})$ \cite{Weber2005, Baccelli06}, for both the
path-loss only and the path-loss plus fading model.


Similar to Shannon capacity, transport capacity has in-built reliability (modulo the SINR model definition), i.e., all counted bits are actually successfully received. Transmission capacity on the other hand follows a fire and forget principle (one-shot transmission), where a packet is guaranteed to be successfully received with probability $1-\epsilon$, but if unsuccessful, there is no penalty or a retransmission procedure.  
Transport capacity definition, however, allows for global arbitrary coordination (scheduling of transmission from all possible transmitters). 

With retransmissions, a generalization of the transmission capacity, called the {\it delay normalized} transmission capacity $C_{\sfR}(\cS)$ with BS strategy $\cS$ \cite{Andrews2009}, is defined as the number of successfully delivered packets in the network under the SINR model subject to a maximum limit on the number of retransmissions $\sfR$. Let $N_\sfR(t)$ be the number of successfully received packets at the MU until time $t$ with at most $\sfR$ retransmissions from the corresponding BS. Then, 
$$C_{\sfR}(\cS) = \lambda \lim_{t\rightarrow \infty} \frac{N_\sfR(t)}{t}.$$

From the renewal reward Theorem \cite{Cox1962}, $C_\sfR$ is also given by
 \begin{equation}
\label{eq:tcmulti}
C_\sfR(\cS)= \frac{ \lambda P_s^\sfR}{\bbE\{D(\sfR)\}}, 
\end{equation}
where $P_s^\sfR$ is the probability that a packet is successfully received at the MU, and $D(\sfR)$ is the random variable denoting the actual number of retransmissions needed, with at most $\sfR$ retransmissions. 

Our capacity definition \eqref{eq:Cdefn} is essentially the limit of $C_\sfR$ as
$\sfR \rightarrow \infty$, where $P_s^\sfR\rightarrow 1$, and yields 
 \begin{equation*}
C(\cS)= \lim_{\sfR \rightarrow \infty}\frac{P_s^\sfR \lambda}{\bbE\{D(\sfR)\}} = \frac{\lambda}{\bbE\{D\}},
\end{equation*} where $D$ is now the absolute delay between any BS $x$ and MU $y$ it is serving with no retransmission limit defined as 
$$D = \min\{t: \SINR_{xy}(t) \ge \beta\}.$$ 

\begin{remark} Using the limit of \eqref{eq:tcmulti}
as $\sfR \rightarrow \infty$ for our capacity definition is actually not a simplification compared to \eqref{eq:tcmulti}, since each BS serves all its MUs in a round-robin fashion, thereby precluding selection of which MU to serve in each slot. Thus, for any particular MU, with or without constraining the maximum number of retransmissions, there is no change in the average number of successfully received packets, since the underlying random process that determines success/failure of packets does not depend on the packet index. Thus, dropping a packet after hitting the constraint $\sfR$ or continuing with it, yields the same average rate of successful packet transmissions (1/expected delay).
\end{remark}

The delay normalized transmission capacity was first defined in \cite{Andrews2009}, and later studied in \cite{VazeTDR2011}. The basic problem in studying the 
delay normalized transmission capacity is the complicated correlation of SINRs across time slots under the PPP assumption on BS locations \cite{Ganti2009}. For simplifying analysis, \cite{Andrews2009} made a limiting assumption that SINRs across time slots are independent. 
A more rigorous approach was taken in \cite{VazeTDR2011}, but yielded limited analytical results because of the mentioned difficulty and the main results were derived only for very-low-BS density regime.

\subsection{Our results}
In this paper, we avoid making any simplifying/limiting assumptions for studying the joint distribution of SINRs across time slots for deriving bounds on the expected delay seen at any MU. The main result of our paper is that, when each BS is only allowed a local strategy or BSs can accomplish only small-scale coordination, upto constants, \footnote{The lower and upper bounds only differ by constants, and all $c_i$'s are constants.}
$$ \min_\cS\bbE\{D\}= \max\left\{\frac{1}{\lambda^{\frac{\alpha}{2}}}, \exp(\lambda)\right\},$$ and consequently the capacity is
$$
C =  \min\left\{\lambda^{\frac{\alpha}{2}+1}, \lambda\exp(-\lambda)\right\}.$$
In Fig. \ref{fig:capacity}, we illustrate the behaviour of the capacity (upto a constant) of the cellular wireless network as a function of the BS density as the blue curve, that tracks the minimum of the two functions $\lambda^{\frac{\alpha}{2}+1}$ and $ \lambda\exp(-\lambda)$.
\begin{figure}
\centering
\includegraphics[width=3.5in]{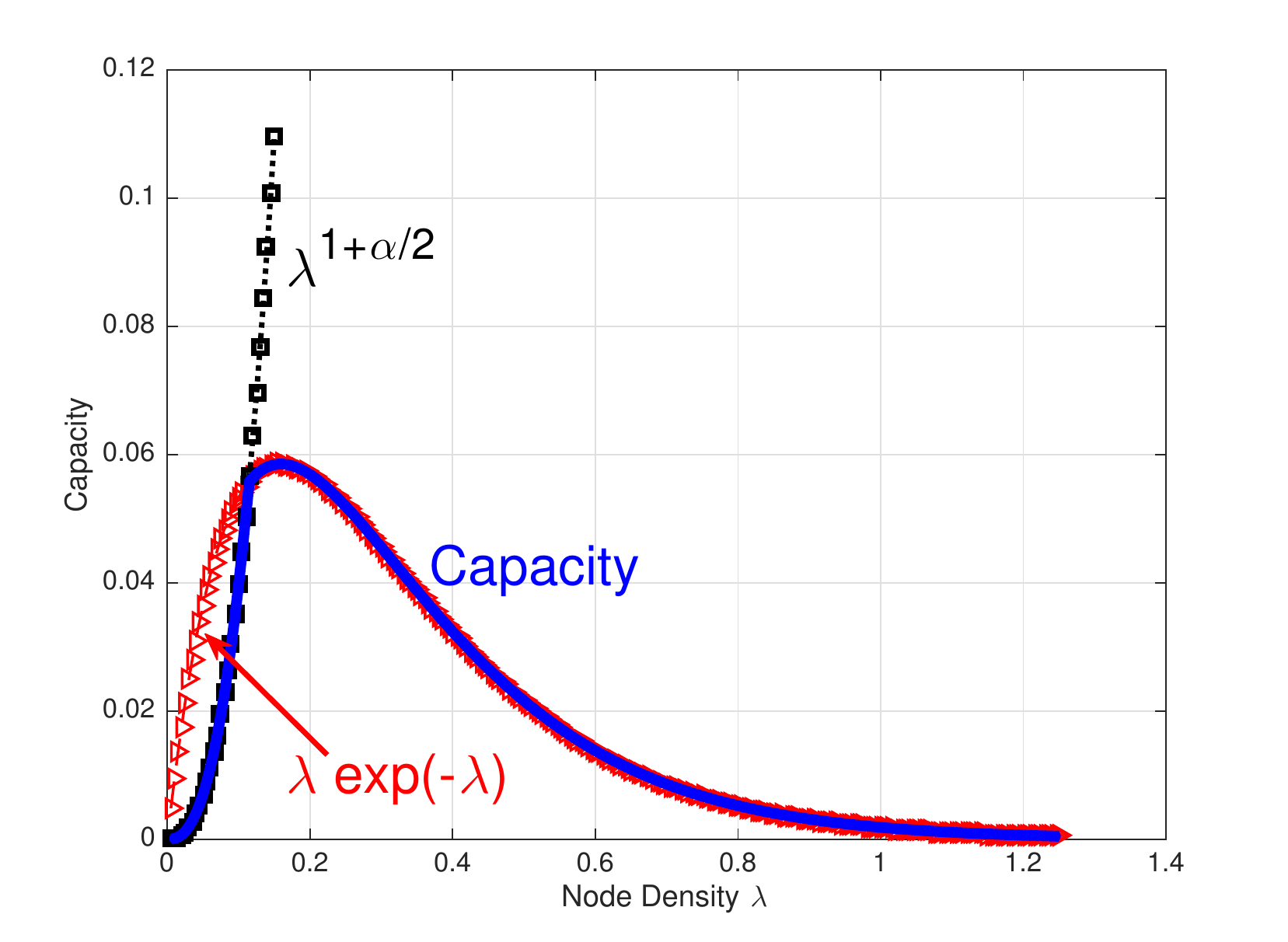}
\caption{Capacity of the Cellular wireless network as a function of the BS density.}
\label{fig:capacity}
\end{figure}
%

Thus, the capacity of a cellular wireless network increases polynomially in the low-BS density regime, then increases as $\lambda\exp(-\lambda)$ in the moderate BS density regime, and eventually starts to decreases exponentially with the density of basestations for local BS strategies or with small-scale BS coordination. 

The main result is derived via the following sub-results.
\begin{itemize}
\item We begin by showing that for any local BS strategy $\cS$
$$C(\cS) = \frac{\lambda}{\bbE\{D\}} \le  \min\left\{\frac{\lambda^{\frac{\alpha}{2}+1}}{c_1}, c_3\lambda\exp(-c_2\lambda)\right\}.$$ It is accomplished via deriving two separate bounds, first
$$\bbE\{D\}\ge \left(\frac{c_1}{\lambda^{\frac{\alpha}{2}}}\right),$$ that is relevant for low-BS density regime, and then $$\bbE\{D\} \ge  \frac{\left(\exp(c_2\lambda)\right)}{c_3},$$ that is tighter in the high-BS density regime.
\item The upper bound on the capacity is shown to hold even if small scale BS coordination is allowed, where a fixed number (independent of the BS density) of BSs can schedule their transmissions jointly. 
\item We show that a simple non-adaptive strategy $\cS$ (first proposed by us in \cite{VazeIyerAAP2016}), where each BS transmits power to completely nullify the path-loss based signal attenuation at the MU it is serving (using the knowledge of the distance to the MU), achieves $$\bbE\{D\} \le  \sqrt{c_5\left(1+ \frac{\Gamma(\alpha+1) }{(\pi \lambda)^\alpha}\right) \exp \left(c_4\lambda\right)},$$
$$\implies 
C(\cS) \ge \min\left\{\frac{\lambda^{\frac{\alpha}{2}+1}}{c_6}, c_7\lambda\exp(-c_3\lambda)\right\},$$
matching the upper bound on the capacity for any strategy $\cS$ upto constants. Importantly, a single policy achieves the capacity upper bound for the both the low and the high-BS density regimes, for which the capacity behaviour is quite different.
\item We show that for a pure ALOHA strategy $\cS$ that is commonly employed and most often studied, $C(\cS) = 0$. Moreover, any ALOHA type strategy, where the serving BS has the knowledge of the 
distance $d$ to its MU, can achieve non-zero capacity, only if its transmit power scales at least as fast as $d^{\alpha-2}$, where $\alpha > 2$ is the path-loss exponent. Since $d$ is a random variable with infinite support, this result directly implies that if each BS has a strict peak power constraint, then its capacity is $0$.
\end{itemize}

Our result shows that in the low-BS density regime the capacity of the cellular wireless network increases polynomially with the BS density, since with low BS density, interference is weak and signal power can be boosted using the proposed power control strategy. In the high-BS density regime, our result shows that the 'realistic' capacity of a cellular wireless network is fundamentally different and significantly smaller than previously thought via the known results on the transport or the transmission capacity (both scale as $\Theta(\sqrt{\lambda})$) in the high-BS density regime. 

The early indication for this negative result in the high BS density regime was visible in \cite{baccelli2011optimal}, that considered an ad hoc network where nodes can be either source/destination/relay in any time slot, rather than taking up a rigid/fixed BS or MU role as in our case. Assuming the location of all nodes 
to be a PPP, under the SINR model, \cite{baccelli2011optimal} showed that the ALOHA strategy has infinite expected delay for a packet to leave its source (to be received successfully at any node in the network, in particular also the nearest node) for any BS density $\lambda$, and consequently has zero effective rate of communication. 

What our upper bound result on the capacity shows that even if each BS is allowed adaptive strategies using local information depending on distance, fading gain or the ack/nacks sent by the MU, the delay  grows at least exponentially in the density of the basestations in the high-BS density regime. 

The proposed strategy to achieve the upper bound only depends on the distance between the BS and the MU, and does not require channel state information (CSI) that changes much faster than the distance. In particular, if the BS-MU distance is $d$ and the path-loss function is $\ell(d)$, then the strategy transmits with power $\ell(d)^{-1}$ with probability $\min\{1,M\ell(d)^{-1}\}$ in any slot dedicated for the MU, where $M$ is the average power constraint. 
Note that the strategy is not  adaptive, making it suitable for practical implementation with low complexity. The basic idea behind the strategy is to transmit infrequently, but whenever an attempt is made, sufficient power is used to compensate the path-loss completely. 
Thus, this strategy limits the overall network interference, while keeping the signal power high enough whenever an attempt is made, and more importantly achieves the capacity upper bound in both the low as well as the high BS density regime, where the capacity behaviour is very different.

One important design implication of our result is in terms of cell densification, where number of basestations are increased to decrease the individual cell-sizes in the hope of improving connectivity and communication rate. 
Single-shot performance metrics such as connection probability or average rate have been considered in past for analyzing the effects of cell densification \cite{zhang2015, chen2012small, samarasinghe2014, baccelli2015scaling, Kountouris2016}, that typically conclude that there is a phase transition; initially performance improves as the BS density increases, then it saturates, and eventually it starts to degrade, but explicit dependence was not identified. Also, it is believed that cell densification can improve the spectral efficiency via frequency reuse.
Taking a more comprehensive view via including retransmissions, where the effects of temporal correlations of SINRs are fully incorporated, we characterize the exact effect of cell densification on the capacity/long-term throughput. We show that in the low-BS density regime, cell densification increases the capacity polynomially, while in the high-BS density regime, BS densification leads to a exponential fall in the capacity. Thus, cell densification needs to be undertaken judiciously. 

\subsection{Comparison with Prior Work}
To compare our capacity metric with transport capacity or transmission capacity (relevant for high BS density), one has to multiply the average distance between the MU and its nearest BS distance $d_0$ (that scales as $ \Theta\left(\frac{1}{\sqrt{\lambda}}\right)$) with the capacity, which again yields $C d_0 = \Theta\left(\exp(-\lambda)\right)$.
The reasons for arriving at such fundamentally different result compared to the transport capacity and the transmission capacity, that both scale as $\Theta\left(\sqrt{\lambda}\right)$, are that transport capacity has stricter reliability requirement (SINR has to be $\ge \beta$ for each successful counted bit on realization basis) but allows for large scale BS/node scheduling, while transmission capacity has one-shot reliability constraint of $1-\ep$, and does not allow for retransmissions. In cellular networks, large scale BS scheduling is not possible, while retransmissions are an integral part of any deployment. The inability of large scale BS scheduling ensures that interference seen at any MU increases with BS density, while with retransmissions, the temporal correlation of SINRs degrades the delay performance.

The positive temporal correlation of SINRs with PPP distributed nodes that was first derived in \cite{Ganti2009}, is the principal reason behind the exponential increase of the expected delay in the high BS density regime. The high positive time correlation in the high-BS density regime ensures that if a packet has not been received till a sufficient number of retransmissions, then increasingly it become less and less likely to succeed, and leads to expected delay being exponential in the BS density. 
We would like to note that the analysis in this paper is non-asymptotic and holds for all $\lambda$, in contrast to the achievability of transport capacity \cite{Franceschetti2004}, that needed the number of nodes to go to infinity for the application of results from percolation theory.


An important result obtained in \cite{andrews2011tractable}, under the exact same model considered in this paper (BS and MU location processes being independent PPPs, and each MU connects to its nearest BS), is that ignoring noise, the one-shot connection probability $P_c = \bbP(\SINR >\beta)$ for any one slot is independent of the BS density $\lambda$. This result pointed towards showing that with increase in $\lambda$, the increase in the signal power because of the decreasing nearest BS distance completely compensates for the increase in the interference, and cell densification can lead to linear increase in network capacity with BS density. 
This conclusion (if valid) will also point towards potentially showing that expected delay $\bbE\{D\}$ is independent of $\lambda$, at least when noise is ignored, contradicting our results.
The real reason for $P_c$ being independent of $\lambda$ under the no-noise condition  \cite{andrews2011tractable}, is actually the use of an {\it amplifying} path-loss model of $\ell(x) = x^{-\alpha}$ in \cite{andrews2011tractable}, that produces unrealistic signal amplification for distances $x <1$. If instead one uses a more realistic path-loss function, e.g., $\ell(x) =\min\{1, x^{-\alpha}\}$ or $\ell(x) =1/(1+ x^\alpha)$ that do not have signal amplification, such a result will not be possible. Thus, in this paper, we consider path-loss function $\ell(x) =\min\{1, x^{-\alpha}\}$, but all our results will apply for any non-amplifying path-loss function, such as $\ell(x) =1/(1+ x^\alpha)$.

\subsection{Tradeoffs} We discuss some of the tradeoffs involved in the context of our results on the capacity of cellular wireless networks. 

\begin{itemize}
\item One integral feature used for improving the performance of a cellular networks has been power control, that is used extensively in CDMA networks, and also finds applicability in other deployments to control the {\it out-of-cell} interference. 
Under the SINR model considered in the paper, all the interference seen at any MU is coming from out of cell. Each BS tries to serve its MU as well as possible, essentially by transmitting with sufficient power which is in conflict with the rate achievable by MUs in other cells. Thus, the choice of transmit power is critical for maximizing the network-wide capacity, where the two conflicting objectives (maintaining sufficient SINR for own users, while minimizing interference for MUs in other cells) have to be balanced out by BSs. With increasing BS density, the number of BSs that are within a fixed distance (that really matter for SINR) from a MU increases linearly, thereby increasing interference rapidly. To counter this  increased interference, the serving BS has to increase its power, but that ends up causing increased interference at out of cell MUs. 

\item Another dimension to the problem is the use of ARQ in cellular networks. With strong positive temporal correlations of SINR, the 
success probability progressively decreases with the retransmission index. Thus, if a packet has not succeeded in few retransmissions, it is likely to take a large number of retransmissions (heavy-tail behavior). The power control issue together with the temporal correlations of SINR, limits the performance of a cellular wireless network, and the expected delay increases at least exponentially with the BS density, leading to capacity falling exponentially with BS density.

\item Path-loss exponent: The path-loss exponent $\alpha$ is an important feature governing the performance in a cellular network. Larger value of $\alpha$ leads to lower interference, but also leads to lower signal power. The path-loss exponent affects the performance depending on the BS density $\lambda$, since with large $\lambda$, the distance between the MU and the BS it connects to decreases, thereby increasing the signal power, but at the same time there are larger number of closeby interferers, even though each interferer contributes lower interference with larger $\alpha$. Thus, how does $\alpha$ impacts the capacity is an important question. We show that the capacity (both upper and lower bound) is indifferent to the value of $\alpha$ for $\lambda >1$, where the BS-MU distance is small and the interference dominates. In the low-BS density regime $\lambda < 1$, however, we show that the capacity decreases with $\alpha$ as $\lambda^{1+\alpha/2}$. Thus, the effect of $\alpha$ is only visible at low-BS densities, where with larger inter-BS distances, interference is weak, and larger $\alpha$ decreases the signal power significantly.  

\item MU density: We have implicitly assumed that the MU density is much higher than the BS density for our capacity definition, which is typically the case in a well-designed system. In case, the MU density is comparable to the BS density, it will lead to some BSs not having any MUs in their Voronoi regions, and consequently becoming inactive. Since the BS and the MU processes are independent, this will lead to an independent thinning of the active BS process (a new PPP with less density). Thus, our results are again applicable in this case. 
\end{itemize}

\subsection{Limitations}
By far, as one can argue, the biggest limitation of our work is the non-information-theoretic definition of capacity. However, as is well known, the Shannon capacity of networks with even 2/3/4 nodes, e.g., interference, broadcast channels, remains unknown. Thus, the hope of finding Shannon capacity of general wireless networks is rather slim. To better understand the fundamental tradeoffs and limitations in networks, 
starting with \cite{Gupta2000}, many simpler but functional definitions of capacity \cite{Weber2005, brian2007capacities} have been proposed and analysed. Our definition is a step in the same direction, with primary focus on cellular networks that are structured, and operate under some well established protocols such as ARQ. 
Even under this simpler definition, the capacity analysis remains non-trivial, and we are able to derive the exact dependence of capacity with the BS density.

Another possible limitation of our work can be argued in terms of assuming independent point processes for the BS and the MU locations, which has become very popular starting with \cite{andrews2011tractable}. One might expect BSs to be deployed in areas where a large number of MUs tend to be present, thus coupling the BS and MU locations process as proposed in the cluster model of \cite{suryaprakash2013stochastic}. We expect our results to be applicable even in this model, since the nearest neighbour distances still have similar distributions, and order-wise results should be similar. The exact results are, however, out of scope of this present paper. 


\section{Notation} For two random variables $X$ and $Y$, $X$ is defined to be stochastically dominated by $Y$ if $\bbP(X\le x) \ge \bbP(Y \le x)$ for all $x$. We will equivalently use stronger or weaker in comparing two random variables, where stronger means the random variable that dominates the other.

\section{System Model}
We consider a cellular network model, that consists of basestations $\{T_n\}$, whose locations are distributed according to a homogeneous PPP $\Phi =\{T_n\}$ with density $\lambda$, popularly called the {\it tractable} model starting from \cite{andrews2011tractable}.
This model is reasonable since with a PPP, given the number of nodes lying in the given area, the node locations are uniformly distributed similar to what is seen practically with many different types of basestations (macro, micro, femto, etc.) overlaid on top of each other.
\begin{figure}
\centering
\includegraphics[width=5.5in]{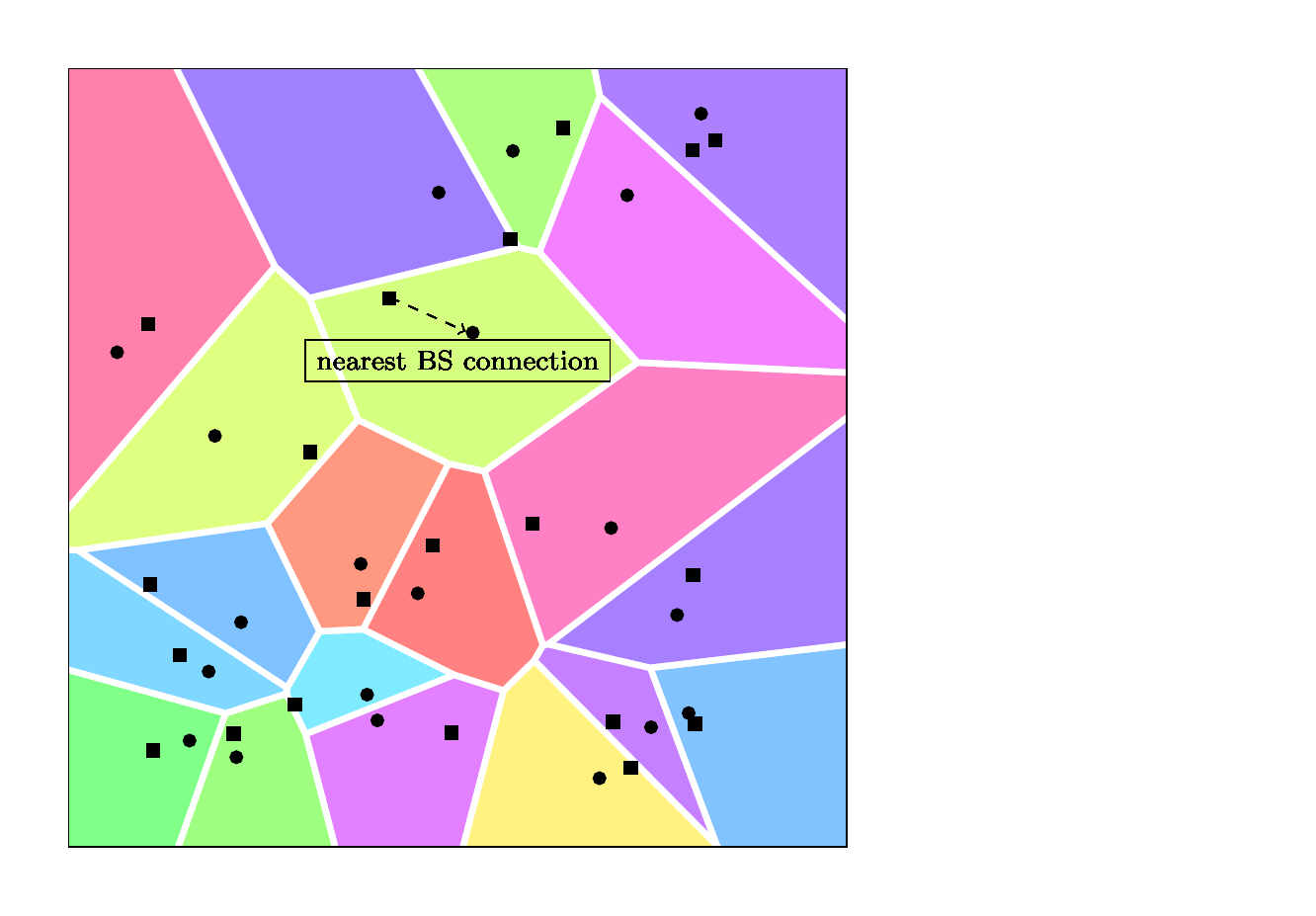}
\caption{Cellular wireless network, where basestations are denoted by circles, and whole area is divided in Voronoi regions, and each mobile (square) node connects to its nearest basestation.}
\label{fig:networkschematic}
\end{figure}

The MUs are also located according to an independent PPP $\Phi_R$ with density $\mu$, and each MU connects to its nearest basestation, which is what is typically the case in a practical deployment. 
Therefore, by the definition of Voronoi regions with respect to basestation locations, all MUs in a Voronoi cell/region connect to its representative basestation.  See Fig. \ref{fig:networkschematic} for the considered network schematic.

We consider a typical MU $\sfm$, and to study the capacity $C$, focus on the expected delay encountered by its packets transmitted from its nearest basestation $\sfn(\sfm)$. 

{\bf Transmission model:} We consider that time is slotted, and each BS transmits to all the MUs that lie in its Voronoi region/cell in a round-robin manner, i.e., equally sharing time slots between them.

{\bf Path-loss model:} We consider the distance based path-loss (signal attenuation) function to be 
$$\ell(d) = \min\{1, d^{-\alpha}\}, \ \text{for} \ \alpha > 2.$$
A simple function $\ell(d) = d^{-\alpha}$ is used widely in literature, however, for $d<1$, 
it produces signal amplification which is unrealizable. The results presented in this paper are valid for most other reasonable path-loss functions such as $\ell(d) = \frac{1}{1+d^{\alpha}}$ that do not have signal amplification. 

{\bf Fading model:} We assume that each node in the network is equipped with a single antenna, and at time $t$, the fading gain between BS $x$ and MU $y$ is denoted by $h_{x,y}(t)$ that is assumed to be exponentially distributed with parameter $1$. Moreover, $h_{x,y}(t)$ is assumed to be independent and identically distributed for all time $t$, and all location pairs $x,y$. 

{\bf Power Constraint:} We assume that each BS has an average power constraint of $M$.

Without loss of generality, we assume that the typical MU $\sfm$ is located at the origin, and the distance to the nearest basestation $\sfn(\sfm)$ from origin is $d_0$, while any other BS 
$z$ is located at distance $z$ (abuse of notation) from the origin. Since we focus on only the MU $\sfm$ located at the origin, we abbreviate $h_{z,0}(t)$ to just $h_z(t)$.
For BS $z$, let $P_z(t)$ be the  transmit power at time $t$, and $\1_z(t)$ be the indicator variable  denoting whether BS $z$ is transmitting at time $t$ or not, with $\bbP(\1_z(t)=1) = p_z(t)$.

The received signal at $\sfm$ located at the origin that is connected to its nearest basestation $\sfn(\sfm)$  is given by 
\begin{align}\nn
y(t) =& \sqrt{P_{\sfn(\sfm)}(t) \ell(d_0) h_{\sfn(\sfm)}(t)} \1_{\sfn(\sfm)}(t) s_t(\sfn(\sfm)) \\ \label{eq:sigmodel}
&+ \sum_{z \in \Phi \backslash \{\sfn(\sfm)\}}\sqrt{ \gamma P_z(t) \ell(z) h_{z}(t)} \1_z(t)s_t(z) + w,
\end{align}
where $s_t(z)$ is the signal transmitted from BS $z$ with power $P_z(t)$ at time $t$, $w$ is the AWGN with variance $\sfN$, 
and $0 < \gamma \le 1$ is the processing gain of the system or the interference
suppression parameter which depends on the transmission/
detection strategy, for example, on the orthogonality between
codes used by different legitimate nodes during simultaneous
transmissions.

Thus, the SINR at $\sfm$ from its nearest BS $\sfn(\sfm)$
in time slot $t$ is given by
\begin{equation}\label{eq:SINR}
\SINR_{\sfn(\sfm), \sfm}(t) = \frac{P_{\sfn(\sfm)}(t) h_{\sfn(\sfm)}(t) \ell(d_0) \1_{\sfn(\sfm)}(t)}{ \ga \sum_{z \in \Phi \backslash \{\sfn(\sfm)\}} P_z(t) \1_z(t)
h_z(t) \ell(z) + \sfN}.
\end{equation}

The transmission from BS $x$ to MU 
$y$ is deemed {\it successful} at time $t$, if $\SINR_{xy}(t) > \beta$,
where $\beta > 0$ is a fixed threshold depending on the rate of information transfer. Let 
\begin{equation}
e_{xy}(t) = 
\begin{cases}
1 & \text{if} \ \SINR_{xy}(t) > \beta, \\
0 & \text{otherwise}.
\end{cases} \end{equation}
Since $h_{(.)}(t)$ is a random variable, multiple
transmissions may be required for a packet to be successfully
received at any node. Thus, a measure of delay, i.e., the number of retransmissions needed to successfully receive packets, is required, which is defined as follows.

\begin{defn}\label{defn:delay} Let the minimum time (delay) taken by any
packet to be successfully received at $\sfm$ located at the origin from its nearest BS $\sfn(\sfm)$ be
$$D = \min\left\{ t >0 : e_{\sfn(\sfm), \sfm}(t) =1  \right\}.$$
\end{defn}

\begin{defn}\label{defn:capacity} As stated earlier, we consider the capacity to be 
$$C = \sup_{\cS} \frac{\lambda}{\bbE\{D\}},$$
where $\cS$ is any BS strategy followed by all BSs.
\end{defn}
Thus, we are fixing BS density $\lambda$ and SINR threshold $\beta$, and looking for the best possible strategy that achieves the capacity. Next, we derive upper bounds on this capacity definition.

\section{Upper Bound on Capacity}
\subsection{ALOHA type strategies}
We begin by considering the most common ALOHA type strategies, where each BS transmits with probability $p$ i.i.d. in each slot with power $P$, such that $pP=M$ to satisfy the average power constraint of $M$. For finding a lower bound on the expected delay, we consider the scenario where  other than BS $\sfn(\sfm)$, there is no other active BS in the network. Thus, the interference seen at $\sfm$ is zero, and we only consider additive noise, in which case $\SINR= \SNR$. Clearly, this will result in lower expected delay. 

Without any interference, there are only two sources of randomness, i) distance $d_0$ between BS $\sfn(\sfm)$ and $\sfm$ that is fixed for all time slots given a realization of $\Phi$ and ii) fading variables $h_{\sfn(\sfm)}(t)$ that are i.i.d. for each time slot $t$. Note that $p$ and $P$ can be arbitrary, however, with ALOHA strategy, $p$ has to be constant across time slots, so it cannot depend on $h_{(.)}(t)$, but can depend on $d_0$. 

We first show that if the transmit power $P$ does not scale as fast as $d_0^{\alpha-2}$, then the expected delay with an ALOHA type strategy is infinity. 
\begin{thm}\label{thm:alohalb} For any ALOHA type strategy, if $P = o(d_0^{\alpha-2})$,  
$$\bbE\{D\} = \infty,$$ and consequently $$C = 0.$$
\end{thm}
\begin{proof} To prove this result, consider that other than BS $\sfn(\sfm)$ there is no other active BS in the network. Thus, the interference seen at $\sfm$ is zero. 
Let the BS $\sfn(\sfm)$ use an ALOHA type policy where it transmits in each slot with probability $p$ i.i.d. with power $P$, such that $pP=M$ to satisfy the average power constraint. 

Given a realization of the BS PPP $\Phi$, the distance $d_0$ is fixed. Thus, conditioned on $d_0$, the event of transmission $\1_\sfn(\sfm)=1$, and the success event $e_{\sfn(\sfm), \sfm}(t) = 1$ $\left(h_{\sfn(\sfm)}(t) > \frac{\beta \sfN \ell^{-1}(d_0)}{P}\right)$ are independent. Moreover, $e_{\sfn(\sfm), \sfm}(t)$ is also independent across $t$ because of i.i.d. assumption on $h_{(.)}(t)$. Recall that with independent trials and per-trial success probability of $q$, the expected time till the earliest success is $1/q$. Hence 
\begin{eqnarray}\nn
\bbE\{D| \Phi\} &=& \frac{1}{p \bbP\left(h_{\sfn(\sfm)}(t) > \frac{\beta \sfN \ell^{-1}(d_0)}{P}\right)}, \\
&=& \frac{1}{p \exp\left(\frac{-\beta \sfN \ell^{-1}(d_0)}{P}\right)},
\end{eqnarray}  
since $h_{(.)}(t) \sim EXP(1)$.
Since $pP \le M$, we have 
\begin{eqnarray}
\bbE\{D| \Phi\} 
&\ge & \frac{P\exp\left(\frac{\beta \sfN \ell^{-1}(d_0)}{P}\right)}{M }.
\end{eqnarray}  
Note that $\ell(d_0) \le d_0^{-\alpha}$. Thus, 
\begin{eqnarray}
\bbE\{D| \Phi\} 
&\ge & \frac{P \exp\left(\frac{\beta \sfN d_0^{\alpha}}{P}\right)}{M }.
\end{eqnarray}  
Hence  
$$\bbE\{D\} = \bbE\{\bbE\{D| \Phi\} \} \ge \frac{ \bbE\left\{P\exp\left(\frac{\beta \sfN  d_0^{\alpha}}{P}\right)\right\}}{M }.$$
From Proposition \ref{prop:nnpdf}, we know that $f_{d_0}(y) = 2 \pi \lambda y \exp\left(-\lambda \pi y^2\right)$. Thus, 
\begin{eqnarray}\nn
\bbE\{D\} 
&\ge & \frac{1}{M}\int_0^\infty P 2 \pi \lambda y \exp\left(-\lambda \sfN \pi y^2\right) \exp\left(\frac{\beta \sfN y^{\alpha}}{P}\right) dy, \\
&=& \infty,
\end{eqnarray} if $P = o(d_0^{\alpha-2})$ where $\alpha >2$.

\end{proof}
Following important corollaries are immediate.
\begin{cor}\label{cor:ALOHA1} With the pure ALOHA strategy, where the BS transmits independently of the distance $d_0$ to the MU it is transmitting to,  $$\bbE\{D\} = \infty.$$
\end{cor}
\begin{cor}\label{cor:ALOHA} If there is a peak power constraint at each BS, then again with the ALOHA strategy,  
$$\bbE\{D\} = \infty.$$
\end{cor}
\begin{proof}
Note that $d_0$ is unbounded, while a peak power constraint will always have $P = o(d_0^{\alpha-2})$, since $\alpha >2$. Hence $\bbE\{D\} = \infty$.
\end{proof}
Corollary \ref{cor:ALOHA1} shows that the pure ALOHA strategy that assumes no knowledge of distance $d_0$, has infinite delay and leads to zero capacity. This is similar to what was shown in \cite{baccelli2011optimal} for a single ad hoc network, that the expected delay for any packet leaving its source successfully is infinite using an ALOHA strategy. Even when an ALOHA strategy can modulate its power depending on distance $d_0$, this result shows that the power has to scale at least as much as $d_0^{\alpha-2}$ to get non-zero capacity. 
Corollary \ref{cor:ALOHA} implies that any BS can only serve MUs that lie within a bounded area in its Voronoi cell under the peak-power constraint using the ALOHA strategy.
\begin{proposition}\label{prop:nnpdf} The cumulative distribution function and probability distribution function of nearest BS distance $d_0$ is 
$$\bbP(d_0 > y) = \exp(-\lambda \pi y^2)$$ and $$f_{d_0}(y) = 2\lambda \pi y \exp(-\lambda \pi y^2).$$
\end{proposition}

We next consider more general (potentially adaptive) strategies, and derive lower bounds on their expected delays. 
\subsection{General Strategies}

\begin{defn}\label{defn:strategy} Let $\cS = \{(p_i,P_i), i\in \cI: p_iP_i \in [M/\tau, M]\}$, where $\tau \ge 1$ is a constant,  be a collection of probability of transmission $p_i$ and power transmission $P_i$ pairs. A {\it strategy} for BS is to choose any element of $\cS$ at any given time slot, 
where the current choice can be adaptive, i.e., it can depend on entire history of earlier choices.\footnote{This strategy can closely emulate any transmission policy used by BS under an average power constraint $M$. 
The lower bound on $p_iP_i \ge M/\tau$ reflects the considered setting of BS density being smaller than the MU density, where each BS has at least one MU in its Voronoi cell.} If $p_s,P_s$ is chosen for a slot, then a BS transmits power $P_s$ with probability $p_s$ in that slot.
 Note that with this definition, the average power constraint of $M$ is satisfied automatically. 
\end{defn}

In this paper, for being as close to practical wireless networks, we begin by restricting ourselves to local strategies that are defined as follows. In the sequel, we consider small-scale BS coordination strategies as well.

\begin{defn} A transmission strategy adopted by a BS for communicating with its MU $\sfm$ is called {\bf local}, if it only depends on either the fading gain or the distance between itself and the connected MU $\sfm$, or the history of success/failure event $e_{\sfn(\sfm), \sfm}$ at $\sfm$ i.e., ack/nack signal sent back by the $\sfm$. In particular, BS $\sfn(\sfm)$, can choose $p_{\sfn(\sfm)}(t), P_{\sfn(\sfm)}(t)$ pair for time slot $t$, depending on $h_{\sfn(\sfm)}(t)$, $d_0$, and history of $e_{\sfn(\sfm), \sfm}(s)$, and $p_{\sfn(\sfm)}(s), P_{\sfn(\sfm)}(s)$ for  $1\le s< t$.
\end{defn}

ALOHA is one simple example of a local strategy. Essentially local strategies preclude global coordination or scheduling across BSs that is very costly and seldom employed in practice in cellular networks. Importantly, local strategies allow adaptation, e.g., power control, that is commonly implemented in practice.

Next, we present the first main result of the paper that upper bounds the capacity of any cellular wireless network, when each BS follows a local strategy.

\begin{thm}\label{thm:lb} For any local strategy $\cS$ employed by a BS, the expected delay at the typical MU $\sfm$ satisfies
\begin{equation}\label{eq:newlb1}
\bbE\{D\}\ge \left(\frac{c_1}{\lambda^{\frac{\alpha}{2}}}\right),
\end{equation} as well as 
\begin{equation}\label{eq:newlb2}
\bbE\{D\} \ge  \frac{\left(\exp(c_2\lambda)\right)}{c_3}.
\end{equation}
Consequently, the network wide capacity is 
 $$C = \frac{\lambda}{\bbE\{D\}} \le  \min\left\{\frac{\lambda^{\frac{\alpha}{2}+1}}{c_1}, c_3\lambda\exp(-c_2\lambda)\right\}.$$ 
 
\end{thm}
The proofs for lower bounds \eqref{eq:newlb1}, \eqref{eq:newlb2} are derived in
Appendix \ref{app:expdelaylb}, and 
Appendix \ref{app:expdelaylb2}, respectively.

Theorem \ref{thm:lb} shows that for low-BS densities, the expected delay decreases at least polynomially with the BS density, where \eqref{eq:newlb1} dominates. In the low-BS density regime, interference is weak, and increasing BS density results in decreased BS-MU distances, boosting the signal power.
For moderate and high BS densities, \eqref{eq:newlb2} dominates, and our result shows that the expected delay grows at least exponentially with the BS density. 

For large BS densities, Theorem \ref{thm:lb}
is essentially a negative result that shows that even when each BS has all the local information, that can be used adaptively, the expected delay increases at least exponentially with the density of BSs. Consequently, the network wide capacity decreases exponentially with the increase in the density of BSs in the high-BS density regime.

Theorem \ref{thm:lb} also suggests that both transport and transmission capacity definitions $\left(\text{both scale as \ $\Theta\left(\sqrt{\lambda}\right)$}\right)$ overestimate the {practically realizable} capacity of a cellular wireless network in the high-BS density regime. The reason for this is that even though with transport capacity, the reliability condition is strict, but large scale BS coordination or scheduling is allowed, while with transmission capacity, the reliability constraint is loose and communication is one-shot. 

We briefly discuss the key ideas used to derive the lower bounds \eqref{eq:newlb1}, \eqref{eq:newlb2} on the expected delay. 
To derive \eqref{eq:newlb1}, we focus on the low-BS density regime, where interference is weak. 
Thus, to derive \eqref{eq:newlb1}, we completely ignore the interference, and hence the SINR seen at $\sfm$ in any time slot $t$ is
\begin{equation}\label{eq:sinrlbspecial1}
\SINR_{\sfn(\sfm), \sfm}(t) = \frac{P_{\sfn(\sfm)} h_{\sfn(\sfm)} \1_{\sfn(\sfm)}(t) \ell(d_0)}{\sfN}.
\end{equation}

One intuitive (but incorrect) way to see the lower bound \eqref{eq:newlb1} is to remove fading $h_{\sfn(\sfm)}\equiv 1$, in which case, success happens ($\SINR > \beta$ \eqref{eq:sinrlbspecial1}) only if $P_{\sfn(\sfm)} > \ell(d_0)^{-1}\beta\sfN$. The average power constraint of $M$, then implies that the probability of transmission $p_{\sfn(\sfm)}(t) = \bbE\{\1_{\sfn(\sfm)}(t)\} \le M (\beta \sfN)^{-1}\ell(d_0) $. Therefore, the expected delay is at least $\bbE\{\frac{1}{p_{\sfn(\sfm)}}\} =  \frac{\beta \sfN}{M}\bbE\{\frac{1}{\ell(d_0)}\}\propto 1/(\pi \lambda)^{\alpha/2}$ from Proposition \ref{prop:nnpdf}. Even though it is tempting to claim that removing fading ($h_{\sfn(\sfm)}\equiv 1$) only decreases the expected delay, however, it is false. 

The main idea to derive the correct proof for  \eqref{eq:newlb1} remains similar, however, instead of removing fading ($h_{\sfn(\sfm)}\equiv 1$), 
we consider a stronger (stochastically dominating) fading distribution than the $EXP(1)$. Thereafter, we condition on $d_0$, and then bound the expected delay when the BS can use the optimal power transmission strategy (Proposition \ref{prop:waterfill}) depending on the 'stronger' fading gains $h'
_{\sfn(\sfm)}$. Finally, the result is obtained by taking the expectation with respect to $d_0$. 
In the high-BS density regime, interference is the key bottleneck and cannot be ignored for finding meaningful lower bounds on the expected delay. The interference seen at any mobile node depends 
on the transmission strategies of other BSs, that can be correlated across time and space.
Thus, the lower bound in the high-BS density is far more complicated to derive, since we have to argue for all possible BS strategies across all BSs.

To lower bound the expected delay in the high-BS density regime, we have to control the joint distribution 
$\bbP\left(\SINR_\sfm(1) < \beta, \dots, \SINR_\sfm(n) < \beta\right)$. Unfortunately, because of potentially correlated BS strategies $\bbP\left(\SINR_\sfm(1) < \beta, \dots, \SINR_\sfm(n) < \beta\right) \ne \prod_{t=1}^n\bbP\left(\SINR_\sfm(t) < \beta\right)$. A common work around (e.g. when BSs use ALOHA type strategies) is to condition on the BS and the MU location processes, $\Phi, \Phi_R$, and then use the conditional independence 
\begin{align}\nn
\bbP\left(\SINR_\sfm(1) < \beta, \dots, \SINR_\sfm(n) < \beta | \Phi, \Phi_R \right) &=  \\  \label{eq:rembott}  \prod_{t=1}^n\bbP\left(\SINR_\sfm(t) < \beta | \Phi, \Phi_R\right).& 
\end{align} Even \eqref{eq:rembott} is not true when the BS strategies are temporally and spatially correlated as is the case here. Thus, we need to make few more enhancements that increase the SINR (and decrease the expected delay) for which \eqref{eq:rembott} is true allowing us analytical tractability on the expected delay.
We make the following three enhancements.
\begin{enumerate}
\item For the typical BS-MU link, we let $\ell(d_0) = 1$. Since $\ell(d_0) \le 1$, clearly, this can only improve the expected delay seen at $\sfm$.
\item For all the non-typical BS $z$, we also let $\ell(d_{zu}) = 1$ for all MUs $u$ being served by $z$. 
\item Completely remove the interference seen at any MU $u$ served by a non-typical BS $z$.
\end{enumerate}

Enhancement (1) ($\ell(d_0) = 1$) is made to completely eliminate the need for the BS strategy to depend on the distance to its MU, since the path-loss function is at most $1$. The possibility of arbitrarily large $d_0$ was the limitation for the ALOHA strategy that was exploited to show its unbounded expected delay in Corollary \ref{cor:ALOHA1}, which is avoided by using $\ell(d_0) = 1$ for deriving a lower bound over all strategies.

Enhancements (2) and (3) increase the signal power and reduce the interference power seen at a non-typical MU $u$, respectively. Thus, to get the same expected delay performance possible without the enhancements (2) and (3), under enhancements (2) and (3) the power transmitted by BS $z$ is stochastically dominated while serving $u$. This in turn, reduces the interference power seen at the typical MU $\sfm$. More importantly, the power transmission strategy of any non-typical BS $z$ does not depend on $d_{zu}$ with $\ell(d_{zu}) = 1$ or any other BS in the network because of no interference. 
Thus, enhancement (2) and (3) together makes the power transmission strategy for all non-typical BSs only depend on the fading gains between BS $z$ and $u$, that are independent across space and time.

Thus, under the three enhancements, we have that
\begin{align*}
\bbP\left(\tilde {\SINR}_\sfm(1) < \beta, \dots, \tilde {\SINR}_\sfm(n) < \beta | \Phi, \Phi_R\right) &= \\ \prod_{t=1}^n\bbP\left(\tilde {\SINR}_\sfm(t) < \beta | \Phi, \Phi_R\right)&,
\end{align*}
where the superscript $\tilde{} \ $ represents the enhanced SINRs.
Using this independence of power transmission strategies from all the BSs that contribute interference at the typical MU,
we then bound the expected delay via the tail probability (joint distribution of consecutive failure events) by considering only the interferers (BSs other than $\sfn(\sfm)$) that lie inside a disc of unit radius around the typical MU $\sfm$. Essentially, as the BS density $\lambda$ increases, the nearest BS distance decreases but the signal power saturates at some level, while the interference keeping growing sufficiently with the increase in the number of interferers in the unit disc, and the temporal correlation of SINRs is strong enough even under these three enhancements.

One question that is of immediate interest is: can the expected delay be decreased if small-scale BS coordination is allowed. We answer this in the negative in the following Corollary.
\begin{cor}\label{cor:bscord} When small scale BS coordination is allowed, i.e., if a fixed number (independent of BS density) of nearest BSs can schedule their transmissions jointly, then for any BS strategy
\begin{equation}\label{eq:newlb1c}\bbE\{D\}\ge \left(\frac{c_1}{\lambda^{\frac{\alpha}{2}}}\right),
\end{equation} as well as 
\begin{equation}\label{eq:newlb2c}\bbE\{D\} \ge  c_9\left(\exp(c_2\lambda)\right).
\end{equation}
Consequently, the network wide capacity is 
 $$C = \frac{\lambda}{\bbE\{D\}} \le  \min\left\{\frac{\lambda^{\frac{\alpha}{2}+1}}{c_1}, \frac{\lambda\exp(-c_2\lambda)}{c_9}\right\}.$$ 
\end{cor}
The proof is provided in Appendix \ref{app:cor:bscord}. 
Lower bound \eqref{eq:newlb1c} is identical to \eqref{eq:newlb1}, since \eqref{eq:newlb1} is derived by completely ignoring interference, and BS coordination has no effect when there is no interference. For the high-BS density regime too, the lower bound \eqref{eq:newlb2c} is similar to \eqref{eq:newlb2} except for constants, since even with coordination (where interference from at most a fixed number of interferers can be canceled), there are sufficient number of interferers in the unit disc as the BS density grows. 

\section{Achievability}
In this section, we consider a simple BS strategy (proposed by us in \cite{VazeIyerAAP2016}) to achieve the capacity upper bound (Theorem \ref{thm:lb}) upto the same order.

{\bf Strategy:} Let for MU $u$, the distance to its nearest BS $\sfn(u)$ be $d_u$. Let each BS know $d_u$ for all the users connected to it. As usual, each BS serves all its connected users in a round-robin fashion by equally splitting the total time or bandwidth. For a slot $t$ designated for a particular MU $u$, BS $\sfn(u)$ transmits with probability $p_{\sfn(u)}(t)$ with power $P_{\sfn(u)}(t)$ given by
\begin{equation}
P_{\sfn(u)}(t) = c \ell^{-1}(d(u)), \qquad p_{\sfn(u)}(t) = M (P_{\sfn(u)}(t))^{-1},
\label{eq:powercontrol}
\end{equation}
where $c = M (1-\ep)^{-1}$, $0< \epsilon <1, \beta \gamma (1-\ep) < 1$ is a constant, and $M = P_{\sfn(u)}(t)p_{\sfn(u)}(t)$ is the average
power constraint. Condition $\beta \gamma (1-\ep) < 1$ is technical and is required for ensuring certain function is $<1$ in \eqref{eqn:bound_cond_exp_int}. Note that $p_{\sfn(u)}(t) \leq 1 - \ep$, since $\ell(d_u) \leq 1$. Thus, in each time slot, with the proposed strategy, each BS makes transmission attempts with transmission power proportional to the
distance to the MU it is serving, to completely nullify the path-loss. The transmission
probability is chosen so as to satisfy the average power constraint of $M$.
It is worthwhile noting that the strategy does not use the knowledge of fading gain $h_{\sfn(\sfm)}(t)$, and is not an adaptive strategy.

\begin{remark} Note that for the proposed strategy \eqref{eq:powercontrol}, the expected power $\bbE\{P_{\sfn(u)}(t)\} < \infty$ and expected probability of transmission $\bbE\{p_{\sfn(u)}(t)\} > 0$ for any BS $n(\sfu)$ from Proposition \ref{prop:nnpdf}.
\end{remark}

From hereon, for analysis purposes, we focus on a typical MU $\sfm$ and its nearest BS $\sfn(\sfm)$, and derive an upper bound on the expected delay seen at $\sfm$.

\begin{thm}\label{thm:ub} The power control strategy \eqref{eq:powercontrol}
achieves the following performance
$$\bbE\{D\} \le \sqrt{c_5 \left(1+ \frac{\Gamma(\alpha+1) }{(\pi \lambda)^\alpha}\right) \exp \left(c_4\lambda\right)},$$ 
where $c_3$ and $c_4$ are constants.
Thus,
$$
C(\cS) \ge \min\left\{\frac{\lambda^{\frac{\alpha}{2}+1}}{c_6}, c_7\lambda\exp(-c_4\lambda)\right\},$$
\end{thm}
Proof is derived in Appendix \ref{app:ach}.
%
Theorem \ref{thm:ub} shows that a simple non-adaptive strategy that does not need to learn the local fading gain is capable of achieving (upto constants) the upper bound  on the capacity $C$. 
 The only local information it needs is the distance $d_0$. Recall that Theorem \ref{thm:alohalb} suggests that any strategy that can achieve non-zero capacity $C$ needs to know $d_0$. Thus, the considered power control strategy is minimal in terms of its learning requirements. Distance $d_0$ can be learned easily via ranging or RSSI measurements and the considered power control strategy can be implemented easily in practice. 

The considered power control strategy completely nullifies the path-loss seen at MU $\sfm$, and makes the received signal power independent of distance $d_0$. Essentially, Theorem \ref{thm:ub} shows that as long as signal attenuation because of the path-loss function is compensated, order-optimal expected delay and capacity $C$ can be achieved. Theorem \ref{thm:ub} shows that a single strategy is sufficient to achieve the capacity upper bound in both the low and high BS density regimes, where the capacity behaviour is quite different.
We would like to note that the proof of Theorem \ref{thm:ub} is 
similar to the one derived in \cite{VazeIyerAAP2016}, however, their the focus was only to show that the expected delay is finite, while here we need the exact dependence of the BS density on the cellular network capacity.

Here we give a back of the envelop calculation for Theorem \ref{thm:ub}. In the low-BS density regime, where the interference is weak, let interference seen at $\sfm$ be a constant $=I$. Then, conditioned on the distance $d_0$ between BS $\sfn(\sfm)$ and $\sfm$, the success probability in any slot with our strategy of $P_{\sfn(\sfm)}(t) = \ell(d_0)^{-1}$ and $p_{\sfn(\sfm)}(t) \approx d_0^{\alpha}$ in the low-BS density regime is
\begin{align*}
p_s & = p_{\sfn(\sfm)}(t) \bbP\left(\frac{h}{\sfN+I} > \beta\right), \\
&= p_{\sfn(\sfm)}(t) \exp\left(-\beta(\sfN+I)\right).
\end{align*}  Thus,
the expected delay seen at $\sfm$ with the proposed power control strategy is $\bbE\{D| d_0\} \propto 1/p_{\sfn(\sfm)}(t) = d_0^{\alpha}$. Taking the expectation with respect to $d_0$, we get $\bbE\{D\} \propto 1/(\pi \lambda)^{\alpha/2}$. In the actual proof, interference is not a constant, and we use Cauchy-Schwarz inequality to decouple the dependence between $p_{\sfn(\sfm)}(t)$ and $I$ to essentially get the same result. 
In the high-BS density regime, mostly $d_0$ is fairly small, and most of the BSs are transmitting for most of the time slots with our strategy. 
Thus, there is sufficient interference, and the intuition for the exponential increase of the expected delay with the BS density is similar to the lower bound derived in Theorem \ref{thm:lb}. One point to note is that we present a unified analysis to obtain the result of Theorem \ref{thm:ub} for both the low and the high-BS density regimes.

\section{Extensions}
We consider two potential extensions of our system model that can possibly provide with better capacity than the results presented in this paper.
\subsection{Multiple Antennas}\label{sec:mimo}
One possible remedy to improve the capacity with increasing BS density is to use multiple antennas 
\cite{Hunter2008, jindal2011multi, Vaze2009} at the BS and the MUs. We next show via a lower bound on the expected delay that the dependence of the expected delay and the capacity remains unchanged with respect to the BS density even with multiple antennas. We consider the case when each BS is equipped with $N$ antennas. Case of MUs also having multiple antennas follows similarly. For sake of brevity, we just indicate the change that multiple antennas can make to the lower bound of Theorem \ref{thm:lb}.

With $N$ antennas at each BS, that are used for beamforming, compared to the SINR expression with single antennas \eqref{eq:SINR}, the only change in the SINR is in terms of the signal power, that is now $|{\bf h}_{\sfn(\sfm)}(t)|$, a chi-square distributed random variable with $2N$ degrees of freedom.

Similar to the technique used to find the lower bound corresponding to low-BS densities in Theorem \ref{thm:lb}, we remove all the interference seen at $\sfm$, to get
\begin{equation}\label{eq:SINRmimospecial}
\SINR_{n(\sfm), \sfm}(t) = \frac{P_{\sfn(\sfm)}(t) |{\bf h}_{\sfn(\sfm)}(t)| \ell(d_0) \1_{\sfn(\sfm)}(t)}{\sfN}.
\end{equation}
Thus, similar to the idea used to derive lower bound \eqref{eq:newlb1}, 
in Theorem \ref{thm:lb}, for $\ell(d_0) \ge \frac{\beta}{M}$, we let  $|{\bf h}| \equiv 1$, and when $\ell(d_0) < \frac{\beta}{M}$, we consider a stronger channel $|{\bf h}'_{\sfn(\sfm)}(t)|$ with PDF $f_{{\bf h}'}(x) = \frac{c}{(N-1)!} x \exp\left(-\frac{x}{N}\right)$\footnote{$c$ is a normalizing constant.} compared to $|{\bf h}_{\sfn(\sfm)}(t)| \sim \chi^2_{2N}$ that has PDF $f_{\bf h}(x) = \frac{1}{(N-1)!} x^{N-1} \exp\left(-x\right)$.
The rest of the proof follows identically to the proof of \eqref{eq:newlb1}, where the resulting lower bound with multiple antennas differs only in constants compared to lower bound \eqref{eq:newlb1}.

Next, to derive bound similar to \eqref{eq:newlb2}, we consider interference, where the SINR at the typical MU $\sfm$ is 
\begin{equation}\label{eq:SINRmimo}
\SINR_{n(\sfm), \sfm}(t) = \frac{P_{\sfn(\sfm)}(t) |{\bf h}_{\sfn(\sfm)}(t)| \ell(d_0) \1_{\sfn(\sfm)}(t)}{ \ga \sum_{z \in \Phi \backslash \{\sfn(\sfm)\}} P_z(t) \1_z(t)
h_z(t) \ell(z) + \sfN}.
\end{equation}

From \eqref{eq:lbint200}, recall that for the $\sfn(\sfm)-\sfm$ link, the tail probability is
\begin{eqnarray*}
\bbP(D > n | \Phi, \Phi_R)
 &=& \prod_{t=1}^n\bbP\left(\SINR_\sfm(t) < \beta \bigl | \Phi, \Phi_R\right),
\end{eqnarray*}
The CDF of a chi-square distributed random variable $X$ with $2N$ degree of freedom is given by 
$\bbP(X \le x) = \int_0^x \frac{s^{N-1}}{(N-1)!}\exp\left(-s\right) ds$. We replace this with a stronger fading channel with PDF $f_{{\bf h}'}(x) = \frac{c}{(N-1)!} x \exp\left(-\frac{x}{N}\right)$.

With $\ell(d_0) = 1$ for the typical user $\sfm$, and dropping additive noise contribution in \eqref{eq:SINRmimo}, as before, we get 
$\bbP\left(\SINR_\sfm(t) < \beta | \Phi, \Phi_R\right)$
\begin{equation} = \bbP\left(\frac{P_{\sfn(\sfm)}(t) {\bf h}' \1_{\sfn(\sfm)}(t)}{ \ga \sum_{z \in \Phi \backslash \{\sfn(\sfm)\}} P_z(t) \1_z(t)
h_z(t) \ell(z)} < \beta \bigl | \Phi, \Phi_R\right).
\end{equation}
From hereon, identical analysis as in the proof of Theorem \ref{thm:lb}, starting from \eqref{eq:lbint200} under restriction $1$ follows, with only difference being in the value of $\delta$. Thus, even with multiple antennas, the expected delay with respect to the BS density remains the same as in Theorem \ref{thm:lb}.

\subsection{Coordinated Multi-Point} 
Another modern technique to improve the performance of cellular wireless networks, is the use of coordinated multi-point (CoMP) \cite{marsch2011coordinated}, where a BS shares its signal with multiple BSs, which then transmit it coherently. Essentially, if $k$ BSs are cooperating using CoMP, then the received signal model at the MU is equivalent to having $k$ transmit antennas at the BS it is connected to without power normalization, and in addition removing the interference from the $k-1$ coordinating other BSs. Thus, CoMP is a combination of having multiple antennas at BS, and small-scale BS coordination. Thus, using results from Corollary \ref{cor:bscord}, and Subsection \ref{sec:mimo}, it easily follows that CoMP also cannot change the order-wise dependence of the BS density on the expected delay and the capacity.

\section{Simulation Results}
In this section, we provide some numerical results to better understand the expected delay performance of cellular wireless networks. 
We consider a BS location process, where the number of BSs is Poisson 
distributed with mean $200$, and each BS lies on a two-dimensional disc of radius $r$. 
To simulate a particular BS density $\lambda$ per m$^2$, the disc radius $r$ is adjusted accordingly, and the typical MU is placed at the center of the disc. 
In all simulations, we use the path-loss exponent  $\alpha =3$,  and SINR threshold $\beta =1$ corresponding to 
$R=1$ bits/sec/Hz, and average power constraint $M=1$. 

In Fig. \ref{fig:pcvsnopc}, we begin by plotting the expected delay of the power control strategy \eqref{eq:powercontrol} in the low-BS density regime. As derived, the expected delay closely follows the polynomial lower bound of $1/\lambda^{\alpha/2}$ in the low-density regime.
In Fig. \ref {fig:pcvsnopcHD}, we plot the expected delay with the power control strategy for higher BS densities, where as expected, the expected delay grows exponentially closely following the derived lower and upper bounds.

In Fig. \ref{fig:BScoop}, we illustrate the expected delay for an ALOHA strategy with BS coordination, where to emulate BS coordination, interference from the nearest $K$ BSs is neglected at the typical MU located at the origin. We see that the delay performance improves, but not order-wise. Finally, in Fig. \ref{fig:comp}, we simulate the performance of CoMP, where for the ALOHA strategy, BS coordination (similar to Fig. \ref{fig:BScoop}) together with multiple antennas at BS is employed. Each BS has $N$ antennas that it uses for beamforming towards its MU. Once again, as discussed earlier, we see that there is no substantial improvement in expected delays with employing multiple antennas at each BS.

\begin{figure}
\centering
\includegraphics[width=3.75in]{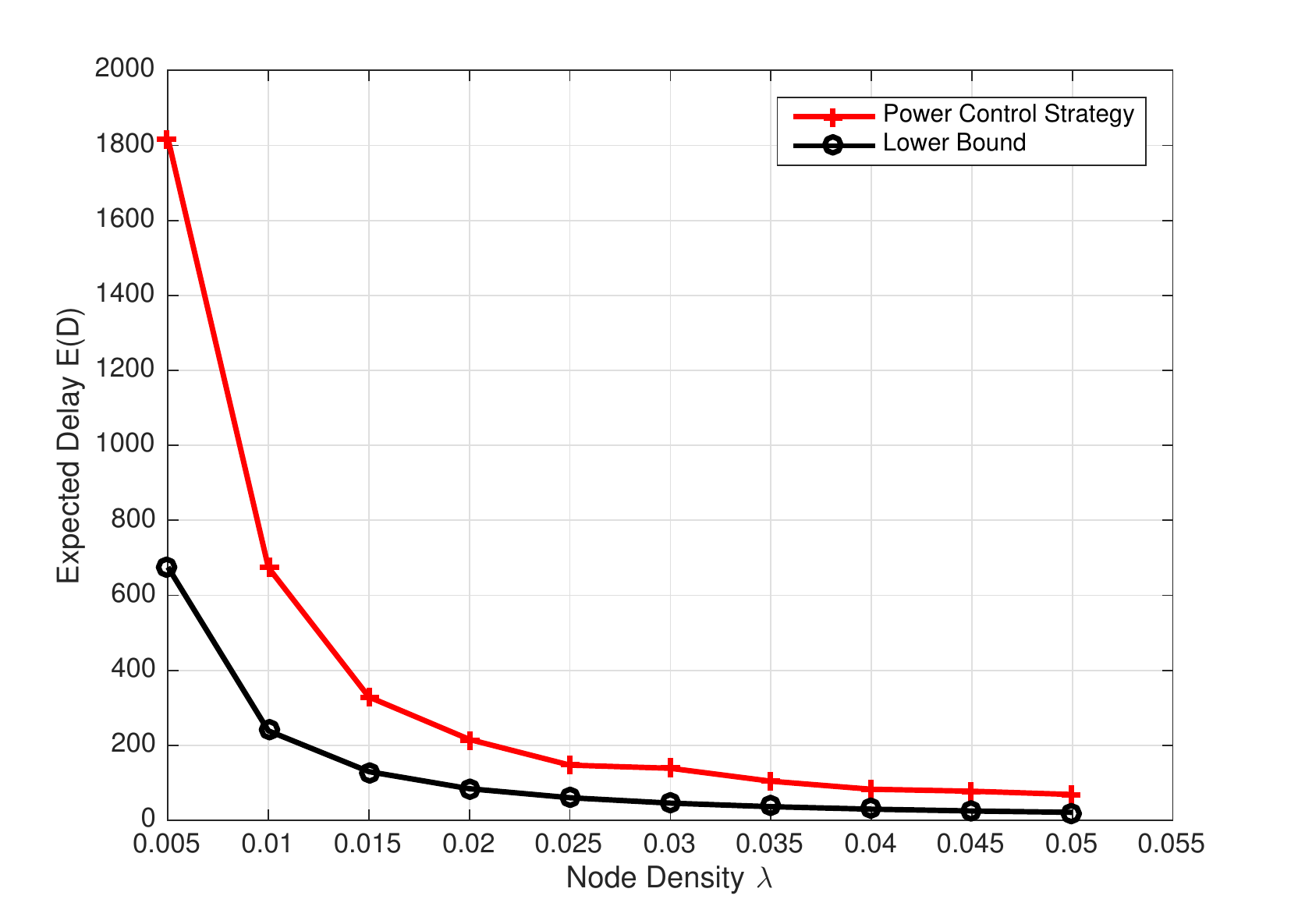}
\caption{Comparison of the expected delay of power control strategy for low-BS densities and the lower bound.}
\label{fig:pcvsnopc}
\end{figure}

\begin{figure}
\centering
\includegraphics[width=3.75in]{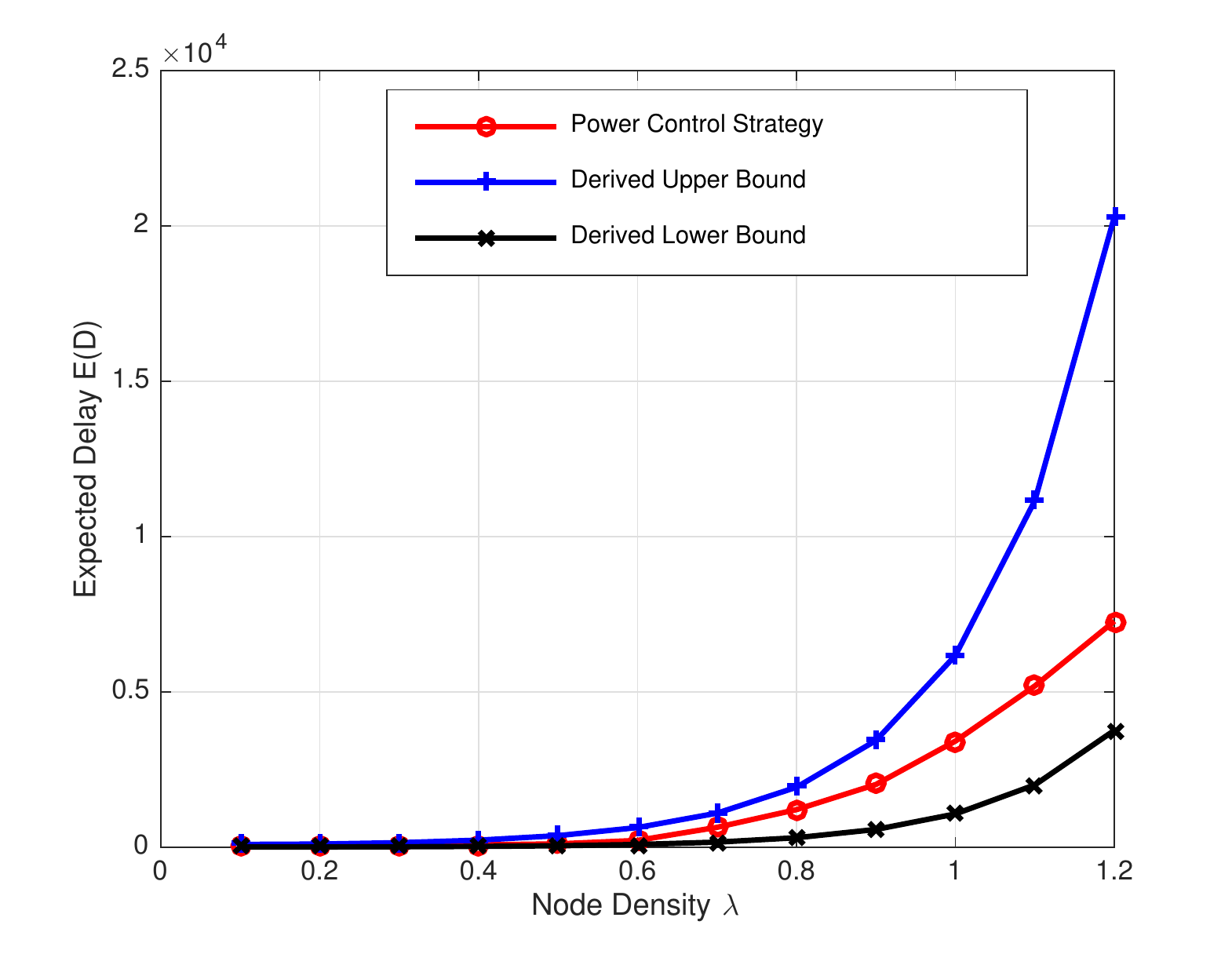}
\caption{Expected delay of the power control strategy for higher BS densities referenced with the lower and the upper bound.}
\label{fig:pcvsnopcHD}
\end{figure}

\begin{figure}
\centering
\includegraphics[width=3.75in]{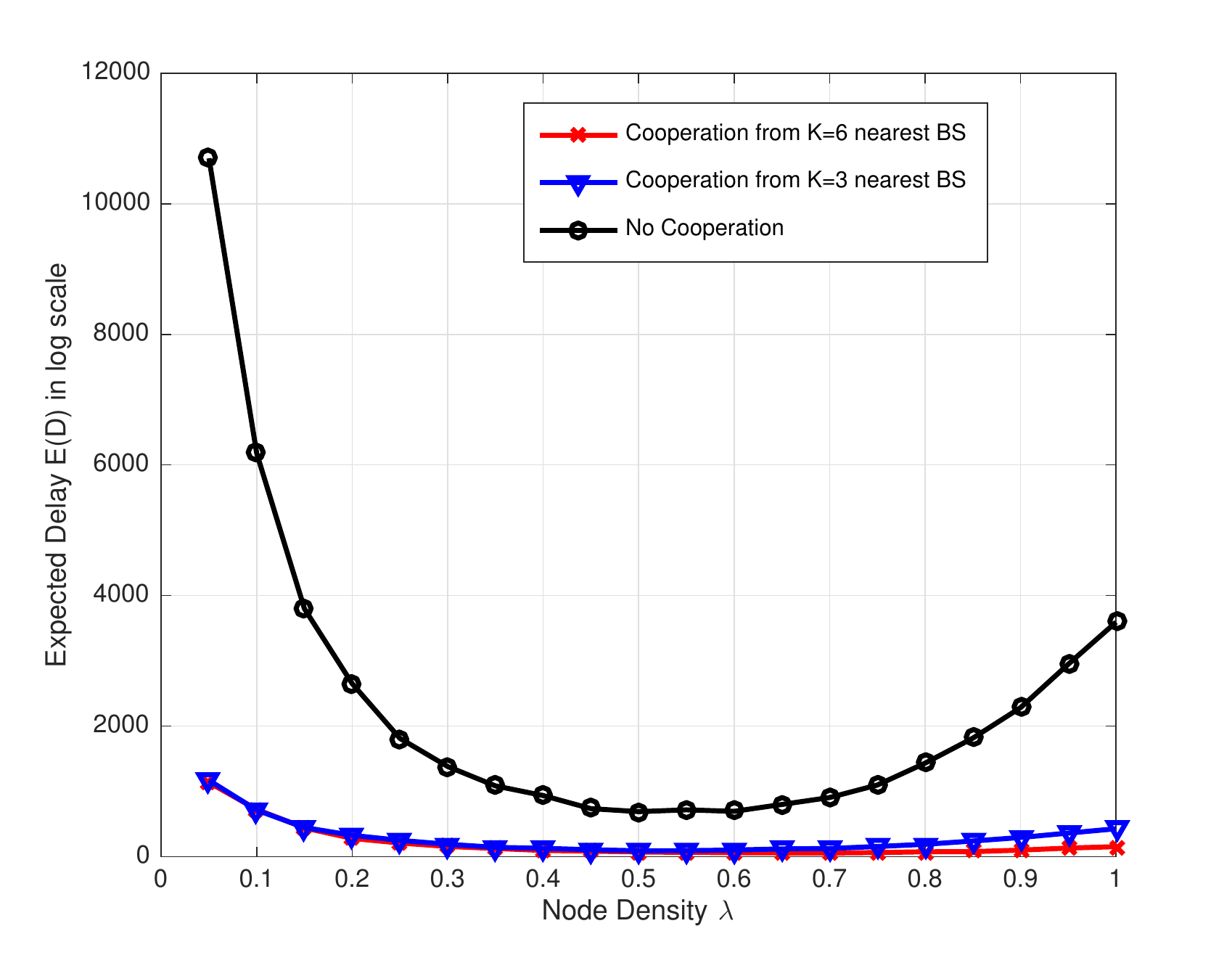}
\caption{Comparison of the expected delay of the ALOHA strategy with BS cooperation, where $7$ closest BSs are not transmitting.}
\label{fig:BScoop}
\end{figure}

\begin{figure}
\centering
\includegraphics[width=3.75in]{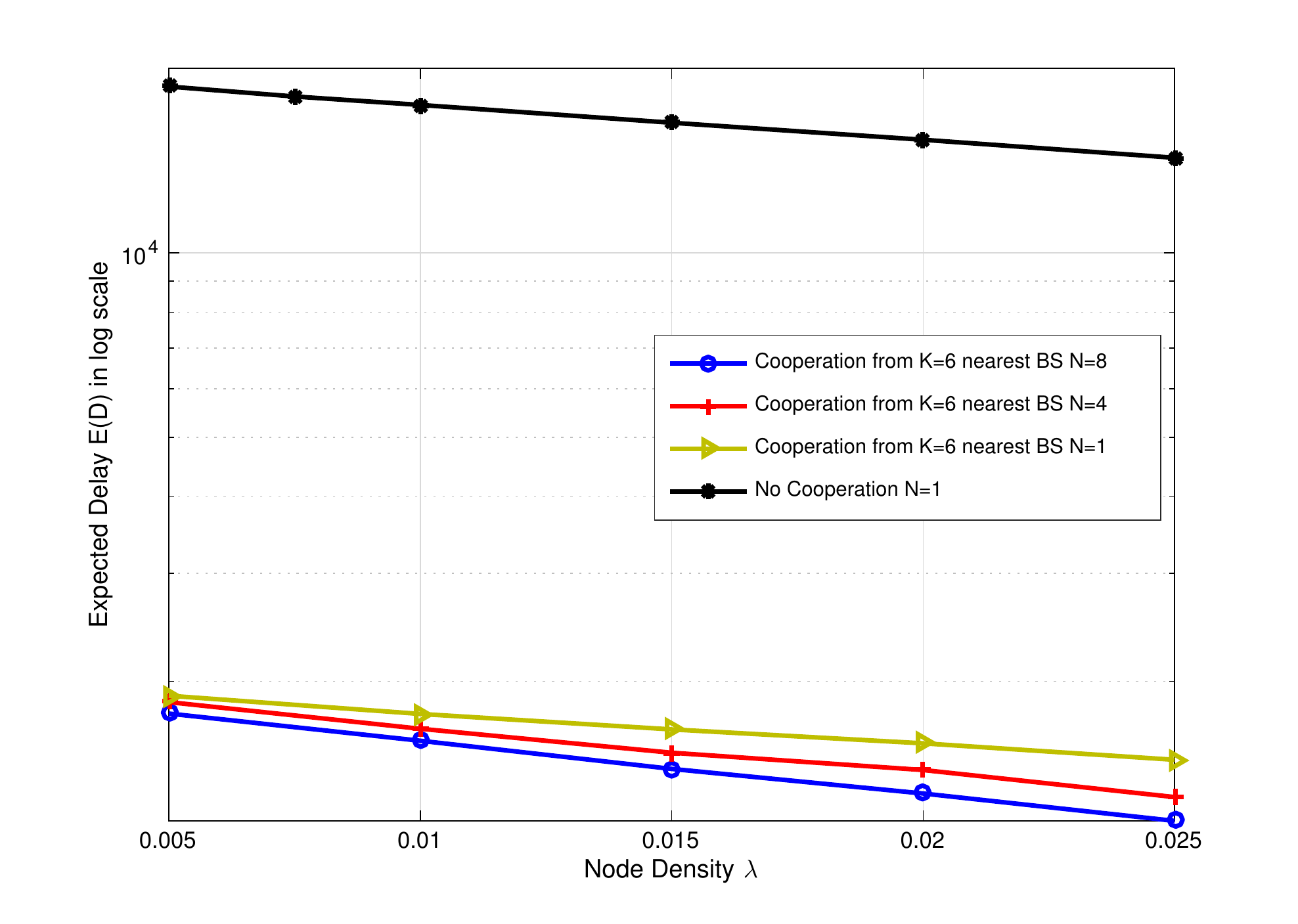}
\caption{Comparison of the expected delay of the ALOHA strategy with CoMP, BS cooperation where $7$ closest BSs are not transmitting, and there are $N$ antennas at the BS.}
\label{fig:comp}
\end{figure}

\section{Conclusion}
In this paper, we derived tight bounds for a natural and practical definition of capacity for cellular wireless networks. 
Even though our capacity metric is weaker than the Shannon capacity, however, it closely matches the throughput measure  
observed in a real-life implementation of cellular wireless networks, where  ARQ is implemented universally. Most importantly, we 
were able to derive the exact dependence of the BS density on the capacity of cellular networks, 
which has generally escaped analytical tractability. Conventional wisdom suggests that there is a advantage in increasing BS density; via
increasing the SINR for the cell-edge users or improving the frequency reuse. Our work, exactly characterizes the capacity behaviour as a function of the BS density. For low-BS densities, the capacity increases polynomially with the BS density, while as BS density is increased further the capacity starts to decrease exponentially.
The prior conclusions were drawn depending on the one-shot connection probability, however, in real-life cellular networks use ARQ. We fully incorporate the effects of ARQ, that are more relevant in the high-BS density regime, where there is a strong positive temporal correlations in SINRs that severely degrades the performance as the BS density is increased. 
Typically, in literature one finds performance (capacity or connection probability) analysis of a fixed BS strategy. To benchmark the performance limits of cellular wireless networks, in this paper, we also upper bounded the capacity 
over all BS strategies that are implementable in practical cellular networks. We also show that the upper bound is achievable by a simple BS strategy, that only depends on the distance between the BS and the MU, that we proposed in our earlier work.

\section{Acknowledgements} 
The paper has benefited immensely because of critical comments and suggestions by Chandra Murthy, Sibi Raj B Pillai, and P.R. Kumar. 
\bibliographystyle{IEEEtran}
\bibliography{IEEEabrv,Research}
\appendices
\section{Proof of Lower Bound on Expected Delay \eqref{eq:newlb1}}\label{app:expdelaylb}
To prove \eqref{eq:newlb1}, we need the following result from \cite{caire1999optimum} on the optimal power allocation to minimize outage probability under an average power constraint.

\begin{proposition}\label{prop:waterfill} \cite{caire1999optimum} Under average power constraint of $M$, to minimize the outage probability $\bbP(h_i P_i \le \gamma \beta)$, where $h_i$ is i.i.d. with PDF $f$, the optimal power allocation 
$$ P_ i = \begin{cases} 
(\ga \beta) / h_i & h_i \ge \delta, \\
0 & h_i < \delta,
\end{cases}
$$
where $\delta$ is chosen to satisfy the average power constraint $\int_{x\ge \delta} P_i(x) f(x) dx \le M$, and where lower average power constraint implies larger $\delta$, and vice versa. 
\end{proposition}

Armed with Proposition \ref{prop:waterfill}, we now proceed to prove \eqref{eq:newlb1}. Consider that other than BS $\sfn(\sfm)$ there is no other active BS in the network. Thus, the interference seen at $\sfm$ is zero. Then the SINR seen at $\sfm$ in any time slot $t$
\begin{equation}\label{eq:sinrlbspecial}
\SINR_{\sfn(\sfm), \sfm}(t) = \frac{P_{\sfn(\sfm)} h_{\sfn(\sfm)} \1_{\sfn(\sfm)}(t) \ell(d_0)}{\sfN}.
\end{equation}
Given a realization of the BS PPP $\Phi$, the distance $d_0$ is fixed. Given $\Phi$, we will find a lower bound on $\bbE\{D| \Phi\}.$ 
With fixed $M, \beta$ for two different cases depending on $d_0$, we will consider stronger fading gains to lower bound the outage probabilities $\bbP(\SINR \le \beta | \Phi)$, and consequently lower bound the expected delay, as follows. For ease of exposition, we absorb $\sfN$ in $\beta = \beta \sfN$.

Case 1: $\ell(d_0) \ge \frac{\beta}{M} $. For this case, we first show that the outage probability $p_o= \bbP(\SINR \le \beta | \Phi)$ (for SINR given by \eqref{eq:sinrlbspecial}) with $h\sim EXP(1)$ is larger than when $h\equiv 1$ always (no fading/line of sight). 
With $h\sim EXP(1)$, from Proposition \ref{prop:waterfill}, for any $M$, the outage probability $p_o= \bbP(\SINR \le \beta | \Phi )$ is at least $1-\exp(\delta)$, where $\delta$ is chosen to satisfy the average power constraint 
$ \int_\delta^\infty \frac{\beta}{ \ell(d_0) x}  \exp(-x) dx = M$. With $h\equiv1$, 
under our restriction $\ell(d_0) \ge \frac{\beta}{M}$ that we began with, just by choosing $P_{\sfn(\sfm)} = \beta \ell(d_0)^{-1}$ always, we get outage probability $p_o= \bbP(\SINR \le \beta | \Phi) =0$, while satisfying the average power constraint of $M$. 

%

Case 2: $\ell(d_0) < \frac{\beta}{M}$.
To derive a lower bound in this case, we replace $h_{\sfn(\sfm)} \sim EXP(1)$ with `stronger' $h'_{\sfn(\sfm)}$ that has PDF $f_{h'_{\sfn(\sfm)}}(x)  = x\exp(-x)$ and CDF $\bbP(h'_{\sfn(\sfm)} < x) =  1-(x+1)\exp(-x)$. With $h_{\sfn(\sfm)} \sim EXP(1)$, the CDF is $\bbP(h_{\sfn(\sfm)} < x) = 1-\exp(-x)$.

With the stronger fading gain $h'_{\sfn(\sfm)}$, conditioned on $d_0$, from Proposition \ref{prop:waterfill}, with an average power constraint of $M$,
the outage probability $p_o = \bbP\left(P_{\sfn(\sfm)} h'_{\sfn(\sfm)} \1_{\sfn(\sfm)}(t) \ell(d_0) \le \beta | \Phi\right)$ is minimized when $P_{\sfn(\sfm)} = \frac{\beta}{\ell(d_0) h'_{\sfn(\sfm)}}$ for $h'_{\sfn(\sfm)} >\delta$ and $P_{\sfn(\sfm)}=0$ otherwise, where $\delta$ is such that the average power constraint 
\begin{equation}\label{eq:dummyxx}
 \int_\delta^\infty \frac{\beta}{ \ell(d_0) x}  x\exp(-x) dx = M
 \end{equation} is satisfied. Solving \eqref{eq:dummyxx}, we get that $\exp(-\delta) = \frac{M  \ell(d_0)}{\beta}$, where recall that $\frac{M  \ell(d_0)}{\beta} < 1$ from our restriction.\footnote{We needed the stronger fading gain so that \eqref{eq:dummyxx} can be solved exactly in terms of $\delta$.}
  Hence, the resulting outage probability, $p_o = \bbP(h'_{\sfn(\sfm)} \le \delta)= 1-\left(1-\ln \frac{M\ell(d_0)}{\beta}\right) \frac{M \ell(d_0)}{\beta}$, and  the success probability $p_s=1-p_o = 
\left(1-\ln \frac{M\ell(d_0)}{\beta}\right) \frac{M \ell(d_0)}{\beta}$.

 Thus, combining case 1 and 2, given $\Phi$, the expected delay is at least 
\begin{align}\nn
\bbE\{D| \Phi\} &\ge \1_{\frac{M  \ell(d_0)}{\beta} \ge 1}1 + \1_{\frac{M  \ell(d_0)}{\beta} < 1}\frac{1}{p_{s}}, \\\label{eq:newubound}
&\ge \1_{\frac{M  \ell(d_0)}{\beta} < 1}\frac{\beta \ell(d_0)^{-1}}{\left(1-\ln \frac{M\ell(d_0)}{\beta}\right)M},
\end{align}
where the first equality follows since given $\Phi$, the success events ($h'_{\sfn(\sfm)} >\delta$) are independent across time slots and $\bbE\{D| \Phi\} = 1/p_s$.
From Proposition \ref{prop:nnpdf}, we know that $f_{d_0}(y) = 2 \pi \lambda y \exp\left(-\lambda \pi y^2\right)$. Recall that $\ell(d_0) = \min\{1, d_0^{-\alpha}\}$, taking the expectation of \eqref{eq:newubound} with respect to $d_0$, we get $\bbE\{D\}$
\begin{align}\label{eq:ubounddummy1}
&\ge \frac{\beta}{M} \left(\int_{\ell(y)<\frac{\beta}{M} }
\frac{\ell(y)^{-1}}{\left(1-\ln \frac{M\ell(y)}{\beta}\right)} 2 \pi \lambda y \exp\left(-\lambda \pi y^2\right)  dy\right), \\ \label{eq:ubound1}
&\ge  \frac{\beta}{M} \frac{ c}{(\pi \lambda)^{\frac{\alpha}{2}}},
\end{align}
where $c$ is a constant. The final inequality is easy to see if the lower limit of the integration in \eqref{eq:ubounddummy1} is $0$ instead of $\ell(y)<\frac{\beta}{M}$. Even with $\ell(y)<\frac{\beta}{M}$ as the lower limit, the same dependence on $\lambda$ follows, but with constants depending on $\beta$ and $M$\footnote{Can be checked by Mathematica.}.

\section{Proof of Lower Bound on Expected Delay \eqref{eq:newlb2}}\label{app:expdelaylb2}
Now we prove the lower bound \eqref{eq:newlb2} on the expected delay while incorporating the interference that will be tight for intermediate and higher values of BS densities.
As defined before, the typical BS $\sfn(\sfm)$ follows a local strategy, i.e., it transmits with power $P_{\sfn(\sfm)}(t)$ with probability $p_{\sfn(\sfm)}(t)$ in any time slot $t$, such that $p_{\sfn(\sfm)}P_{\sfn(\sfm)} \in [M/\tau, M]\}$, where $\tau \ge 1$ is a constant, where $(p_{\sfn(\sfm)}(t), P_{\sfn(\sfm)}(t))$ can depend on $d_0$, $h_{\sfn(\sfm)}(t)$, and history $e_{\sfn(\sfm), \sfm}(s)$,  $p_{\sfn(\sfm)}(s), p_{\sfn(\sfm)}(s), 1\le s< t$.

To obtain a lower bound on the expected delay, we set $\ell(d_0) = 1$ for the typical user $\sfm$, to maximize the signal power in terms of path-loss, since $\ell(.) \le 1$. This will allow us to remove the dependence of $(p_{\sfn(\sfm)}(t), P_{\sfn(\sfm)}(t))$ on $d_0$. 
Note that we are not putting any restriction on $d_0$, since that would impact the interference term in the SINR expression.

Let $\cU_z$ be the set of MUs connected to the BS $z \ne \sfn(\sfm)$. As before, we consider the practical setting where the MU density is much larger than the BS density, and $|\cU_z| \ge 1$ for all BSs $z$. \footnote{Otherwise, since MU and BS processes are independent, we get a thinned BS process, for which the derived results apply directly.} For BSs other than $\sfn(\sfm)$, we let all the MUs $u\in\cU_z$ connected to BS $z$ be at most a unit distance away from $z$, i.e., $d_{zu} \le1$, $\forall \ z,u\in\cU_z$, which implies $\ell(d_{zu})=1$. This minimizes the path-loss between $z$ and $u \in \cU_z$, and gives the largest gain $\ell(d_{zu})=1$ for each BS-MU pair ($z,u$) for $u\in\cU_z$. Moreover, we also let that each $u \in \cU_z$ receives no interference from any basestation $z' \ne z$. Let this be called {\bf restriction} $1$ that applies to all 
non-typical BSs $z\ne \sfn(\sfm)$ and its connected users in $\cU_z$. Thus, the SINR seen at any MU $u\in\cU_z, z \ne \sfn(\sfm)$ under restriction $1$ is 

\begin{equation}\label{eq:SINRz}
\SINR_{z, u}(t) = \frac{P_{z}(t) h_{zu}(t) \1_{z}(t)}{  \sfN}.
\end{equation}
As a function of time, let $\tilde {p}_{z}(t)$ and ${\tilde P}_{z}(t)$ (where $\tilde {p}_{z}(t){\tilde P}_{z}(t) \in [M/\tau, M]\}$ from strategy Definition \ref{defn:strategy}), be the power profile used by BS $z$ under restriction $1$ (SINR given by \eqref{eq:SINRz}) to get the same delay/capacity without restriction $1$ while using the optimal (unknown) power profile $p_{z}(t)$ and $P_{z}(t)$.
Clearly, since there is no path-loss and no interference seen at any $u\in \cU_z$ with restriction $1$ \eqref{eq:SINRz}, 
power profile $\tilde {p}_{z}(t)$ and ${\tilde P}_{z}(t)$ is stochastically dominated by $p_{z}(t)$ and $P_{z}(t)$. Equivalently, the interference seen at the typical MU $\sfm$ from BS $z \ne \sfn(\sfm)$ is stochastically dominated with restriction $1$ compared to without having restriction $1$.\footnote{Essentially, what we are doing is that let $S^\star = (p(t),P(t))$ be the (unknown) optimal strategy (power profile) that achieves minimum delay $\bbE\{D^\star\}$ for each BS $z$. Then, under restriction $1$, the power profile needed ${\tilde S}_z=(\tilde {p}_{z}(t), \tilde {P}_{z}(t))$ for all BSs $z\ne \sfn(\sfm)$ to achieve $\bbE\{D^\star\}$ at their respective MUs, is stochastically dominated by $S^\star$. Thus, with each BS using ${\tilde S}_z$, the expected delay at the typical MU $\sfm$ is at most $\bbE\{D^\star\}$.}

Recall that the SINR \eqref{eq:SINR} seen at the typical user $\sfm$ critically depends on the power profile ($p_{z}(t)$ and $P_{z}(t)$) of BS $z \in \Phi \backslash \{\sfn(\sfm)\}$. Thus, using the stochastically dominated power profile $\tilde {p}_{z}(t)$ and ${\tilde P}_{z}(t)$ with restriction $1$ for each BS $z \ne \sfn(\sfm)$ cannot increase the delay at $\sfm$.  Thus, we work under restriction $1$ to lower bound the expected delay seen at $\sfm$.

It is important to note that under restriction $1$, the power profile of BS $z$, i.e.,  $\tilde {\1}_{z}(t)$ and ${\tilde P}_{z}(t)$, only depends on $h_{zu}(t)$. Since $h_{zu}(t)$'s are independent for each BS-MU pair ($z,u$) $u\in \cU_z$ and across time $t$, restriction $1$ helps simplify the ensuing analysis significantly.

We also also neglect additive noise in the SINR expression \eqref{eq:SINR} for deriving the lower bound on the expected delay. 
 
Under this setup ($\ell(d_0) = 1$, no noise, and restriction $1$ for BSs $z \ne \sfn(\sfm)$), the effective SINR, $\tilde {\SINR}$ is similar to \eqref{eq:SINR} for the $\sfn(\sfm)-\sfm$ link and given by
\begin{equation}\label{eq:SINRrest}
\tilde {\SINR}_{\sfn(\sfm), \sfm}(t) = \frac{P_{\sfn(\sfm)}(t) h_{\sfn(\sfm)}(t) \1_{\sfn(\sfm)}(t)}{ \ga \sum_{z \in \Phi \backslash \{\sfn(\sfm)\}} \tilde {P}_z(t) \tilde {\1}_z(t)
h_z(t) \ell(z)}.
\end{equation}
Note that we have kept the path-loss $\ell(z)$ from BS $z$ to the typical MU $\sfm$ as it is with restriction $1$.

 Thus, we have the tail probability for the delay at the typical MU $\sfm$ as $\bbP(D > n | \Phi, \Phi_R)$
\begin{eqnarray}\nn
 & = & \bbP\left(\tilde {\SINR}_\sfm(1) < \beta, \dots, \tilde {\SINR}_\sfm(n) < \beta \right | \Phi, \Phi_R), \\\label{eq:lbint200}
 &=& \prod_{t=1}^n\bbP\left(\tilde {\SINR}_\sfm(t) < \beta \bigl | \Phi, \Phi_R\right),
\end{eqnarray}
where the second equality follows since given $\Phi, \Phi_R$, SINRs are independent; under restriction $1$, 
$\tilde {\1}_{z}(t)$ and ${\tilde P}_{z}(t)$ only depend on $h_{zu}(t)$ that are independent across time, and 
$\1_{\sfn(\sfm)}(t)$ and $P_{\sfn(\sfm)}(t)$ only depend on $h_{\sfn(\sfm)}(t)$\footnote{Under $\Phi, \Phi_R$ interference terms are independent across time and hence previous failures/successes have no effect on $\1_{\sfn(\sfm)}(t)$ and $P_{\sfn(\sfm)}(t)$.} that are independent across time, and $h_z(t)$ is also independent. 



In the Appendix \ref{app:ub}, we present the simpler proof for the case when the fading gain 
$h_{\sfn(\sfm)}(t)$ is not known at the typical BS, popularly called as the channel state information only at the receiver (CSIR). We proceed with the more challenging channel state information at both the transmitter and the receiver (CSIT) case as follows.

From \eqref{eq:SINRrest},
$\bbP\left(\tilde {\SINR}_\sfm(t) < \beta | \Phi, \Phi_R\right)$
\begin{equation} = \bbP\left(\frac{P_{\sfn(\sfm)}(t) h_{\sfn(\sfm)}(t) \1_{\sfn(\sfm)}(t)}{ \ga \sum_{z \in \Phi \backslash \{\sfn(\sfm)\}} \tilde {P}_z(t) \tilde{\1}_z(t)
h_z(t) \ell(z)} < \beta \bigl | \Phi, \Phi_R\right). \label{eq:lbint1}
\end{equation}
Note that with CSIT, $\1_{\sfn(\sfm)}(t)$ and $P_{\sfn(\sfm)}(t)$ can depend on $h_{\sfn(\sfm)}(t)$, and thus we use Proposition \ref{prop:waterfill}. \footnote{$\tilde{p}_{z}(t)$, $\tilde {P}_{z}(t)$ and $h_z(t)$'s are independent and are not available at the BS.}

Under restriction $1$, both $\tilde {P}_z(t)$ and $\tilde{\1}_z(t)$ are independent across time, and are unknown at $\sfm$. Moreover, their distributions are also unknown. Thus, 
to derive a lower bound, we let BS $\sfn(\sfm)$ choose the optimal transmission policy to minimize 
$\bbP\left(P_{\sfn(\sfm)}(t) h_{\sfn(\sfm)}(t) \1_{\sfn(\sfm)}(t) < \ga \beta \right)$ from Proposition \ref{prop:waterfill} that only depends on $h_{\sfn(\sfm)}(t)$ and $\gamma, \beta$.
Let $\tilde {I}(t) = \sum_{z \in \Phi \backslash \{\sfn(\sfm)\}} \tilde {P}_z(t) \tilde{\1}_z(t)
h_z(t) \ell(z)$.

From the structure of the optimal policy in Proposition \ref{prop:waterfill}, the outage probability $\bbP\left(\tilde {\SINR}_\sfm(t) < \beta | \Phi, \Phi_R\right)$
\begin{align}\nn
  =&  \bbP\left(h_{\sfn(\sfm)}(t)< \delta\right) 
\bbP\left(\frac{0}{\tilde {I}(t)} <  \ga \beta \bigl | \Phi, \Phi_R\right) \\\nn
  &+\bbP\left(h_{\sfn(\sfm)}(t)\ge \delta\right)\bbP\left(\frac{\ga \beta}{\tilde {I}(t)} < \ga \beta \bigl | \Phi, \Phi_R\right), \\\nn
= &\bbP\left(h_{\sfn(\sfm)}(t)< \delta\right) + \bbP\left(h_{\sfn(\sfm)}(t)\ge \delta\right) \bbP\left(\tilde {I}(t) >1 \bigl | \Phi, \Phi_R\right), \\ \label{eq:newwfmargpout} 
 = &(1-\exp(-\delta)) + \exp(-\delta) \bbP\left(\tilde {I}(t) >1 \bigl | \Phi, \Phi_R\right),
\end{align} 
where the last inequality follows since $h_{\sfn(\sfm)}(t) \sim EXP(1)$.
Now we derive bounds on $\bbP\left(\tilde {I}(t) >1 \bigl | \Phi, \Phi_R\right)$ to derive a lower bound on the expected delay.

Recall that $\bbP\left(\tilde {I}(t) >1 \bigl | \Phi, \Phi_R\right)$
\begin{align*}
 & = \bbP\left( \sum_{z \in \Phi \backslash \{\sfn(\sfm)\}} \tilde {P}_z(t) \tilde{\1}_z(t)
h_z(t) \ell(z) >1 \bigl | \Phi, \Phi_R\right),\\
& \ge 1- \prod_{z \in \Phi \backslash \{\sfn(\sfm)\}}  \bbP\left( \tilde {P}_z(t) \tilde{\1}_z(t)
h_z(t) \ell(z) < 1 \bigl | \Phi, \Phi_R\right), \\
& = 1- \prod_{z \in \Phi \backslash \{\sfn(\sfm)\}} \bbE\left\{\left( 1-\exp\left(\frac{-1}{\tilde {P}_z(t) \tilde{\1}_z(t) \ell(z)}\right)\right) \bigl | \Phi, \Phi_R\right\},
\end{align*} 
where the second equality follows since $h_z(t)$'s are independent for $z,t$, and $\tilde {P}_z(t), \tilde{\1}_z(t)$ are independent for $z$ under restriction $1$, and the last inequality follows by taking the expectation with respect to $h_z(t)\sim EXP(1)$. 

Taking the expectation with respect to $\tilde{\1}_z(t)$ (Bernoulli with $\bbE\{\tilde{\1}_z(t)\} = \tilde {p}_z(t)$), we get 
$\bbP\left(\tilde {I}(t) >1 \bigl | \Phi, \Phi_R\right)$
\begin{align}\nn
 & = 1- \prod_{z \in \Phi \backslash \{\sfn(\sfm)\}}\bbE\left\{ \left( 1-\tilde {p}_z(t) \exp\left(\frac{-1}{ \tilde {P}_z(t) \ell(z)}\right)\right)| \Phi, \Phi_R\right\}, \\ \nn
 & \ge 1- \prod_{z \in \Phi \backslash \{\sfn(\sfm)\}} \bbE\left\{\left( 1- \tilde {p}_z(t)\exp\left(\frac{-1}{\frac{M/\tau}{\tilde {p}_z(t)} \ell(z)}\right)\right)| \Phi, \Phi_R\right\}, \\ \nn
& = 1- \prod_{z \in \Phi \backslash \{\sfn(\sfm)\}} \bbE\left\{\left( 1- \tilde {p}_z(t)\exp\left(\frac{-\tilde {p}_z(t)}{(M/\tau) \ell(z)}\right)\right)| \Phi, \Phi_R\right\}, 
 \\  \label{eq:finalI0}
 & \ge 1- \prod_{z \in \Phi \backslash \{\sfn(\sfm)\}}\left( 1-  \bbE\left\{\tilde {p}_z(t) | \Phi, \Phi_R\right\}\exp\left(\frac{-1}{(M/\tau) \ell(z)}\right)\right), 
\end{align} 
where the second inequality follows since  $\tilde{p}_z(t)\tilde{P}_z(t) \ge M/\tau$ for each BS $z$, and the final inequality follows since $\tilde{p}_z(t)\le 1$, and the fact that $\tilde{p}_z(t)$ are independent for different BSs $z$ under restriction $1$.

\begin{remark} Recall that the BS density is much smaller than MU density, hence $|\cU_z| >> 1$, i.e., each BS $z$ is transmitting to multiple MUs in its Voronoi cell. 
Thus, for ${\tilde p}_z(t)$ (that is independent across $z$ and $t$) that only depends on fading gains from $z$ to $\cU_z$, we let $\bbE\{{\tilde p}_z(t) | \Phi, \Phi_R\} \ge \eta >0$ (where $\eta$ is a constant) where the expectation is with respect to the fading gains. Essentially, we have that the probability of transmission of any BS is bounded away from zero for any BS deployment (realization of $\Phi$), which is reasonable to assume in practice. 
\end{remark} 
Thus, it follows that 
$\bbP\left(\tilde {I}(t) >1 \bigl | \Phi, \Phi_R\right)$
\begin{align}\label{eq:finalI1}
  & = 1- \prod_{z \in \Phi \backslash \{\sfn(\sfm)\}} \left( 1-\eta \exp\left(\frac{-1}{(M/\tau) \ell(z)}\right)\right).
\end{align} 

Substituting \eqref{eq:finalI1} into \eqref{eq:newwfmargpout}, we get
$\bbP\left(\tilde {\SINR}_\sfm(t) < \beta | \Phi, \Phi_R\right)$
\begin{align*} 
&\ge (1-\exp(-\delta)) \\
 &+ \exp(-\delta)\left( 1- \prod_{z \in \Phi \backslash \{\sfn(\sfm)\}} \left( 1- \eta \exp\left(\frac{-1}{(M/\tau)\ell(z)}\right)\right)\right), \\
 &= \left( 1- \exp(-\delta)\prod_{z \in \Phi \backslash \{\sfn(\sfm)\}} \left( 1- \eta \exp\left(\frac{-1}{(M/\tau)\ell(z)}\right)\right)\right),
\end{align*} which 
from \eqref{eq:lbint200}, gives $\bbP(D > n | \Phi, \Phi_R)$
\begin{align}\label{eq:eq3y}
 &\ge  \prod_{t=1}^n  \left(1- \exp(-\delta)\prod_{z \in \Phi \backslash \{\sfn(\sfm)\}} \left( 1-\eta\exp\left(\frac{-1}{(M/\tau)\ell(z)}\right)\right)\right).
\end{align}

Let $g= \prod_{z \in \Phi \backslash \{\sfn(\sfm)\}} \left( 1-\eta\exp\left(\frac{-1}{(M/\tau)\ell(z)}\right)\right)$. Then $\bbE\{D | \Phi, \Phi_R\}$
\begin{eqnarray*} 
&= & \sum_{n=0}^\infty \bbP(D > n | \Phi, \Phi_R), \\
&\ge & \sum_{n=0}^\infty (1-\exp(-\delta)g)^n,\ \text{from}\ \eqref{eq:eq3y}\\
&=& \frac{1}{1-(1-\exp(-\delta)g)}, \\
&=& \frac{1}{\exp(-\delta) g}.
\end{eqnarray*}
Thus, using the definition of $g$,
$$\bbE\{D | \Phi, \Phi_R\}  \ge \exp(\delta) \prod_{z \in \Phi \backslash \{\sfn(\sfm)\}} \frac{1}{\left( 1-\eta\exp\left( \frac{-1}{(M/\tau)\ell(z)}\right)\right)}.$$

To find $\bbE\{D\}$ from $\bbE\{D | \Phi, \Phi_R\}$, we first use the probability generating functional (Proposition \ref{prop:PGF}) for the PPP $\Phi \backslash \{\sfn(\sfm)\}$ to get $\bbE\{D | d_0\}$
\begin{small}
\begin{eqnarray*} 
&\ge & \exp(\delta)\\
&&\exp\left( \lambda  \int_{\bbR^2 \backslash \bB(0, d_0)}\left( \frac{1}{\left( 1-\eta\exp\left( \frac{-1}{(M/\tau)\ell(z)}\right)\right)}-1\right) dz\right) ,\\
&=&  \exp(\delta) \exp\left( 2 \pi \lambda  \int_{\bbR, z>d_0} \frac{\eta\exp\left(\frac{-1}{(M/\tau)\ell(z)}\right)}
{\left( 1-\eta\exp\left(\frac{-1}{(M/\tau)\ell(z)}\right)\right)} z \ dz\right).
\end{eqnarray*}
\end{small}

Taking the expectation with respect to the nearest BS distance $d_0$, 
whose PDF $f_{d_0}(x)$ is given by Proposition \ref{prop:nnpdf}, we get $\bbE\{D\}$
\begin{eqnarray*} 
&\ge & \exp(\delta)\\
&&\int_x \exp\left(2 \pi \lambda \int_{\bbR, z>x} \frac{\eta\exp\left(\frac{-1}{(M/\tau)\ell(z)}\right)}
{\left( 1-\eta\exp\left(\frac{-1}{(M/\tau)\ell(z)}\right)\right)} z \ dz   \right)\\
&& f_{d_0}(x) dx
\end{eqnarray*}

Splitting the integral into two parts for $x \le 1$ and $x > 1$, and keeping only the part with $x\le 1$, we get $\bbE\{D\}$
\begin{eqnarray*} 
&\ge & \exp(\delta)\\
&&  \int_{x\le 1} \exp\left(2 \pi \lambda \int_{\bbR, z>x} \frac{\eta\exp\left(\frac{-1}{(M/\tau)\ell(z)}\right)}
{\left( 1-\eta\exp\left(\frac{-1}{(M/\tau)\ell(z)}\right)\right)} z \ dz   \right)\\
&& f_{d_0}(x) dx.\end{eqnarray*}
For $z < 1$, $\ell(z) = 1$, thus, $\bbE\{D\}$
\begin{align} \nn 
&\ge  \int_{x\le 1}\exp(\delta)\exp\left(2 \pi \lambda c_2  \int_{x}^1 z \ dz   \right)f_{d_0}(x) dx, \\ \label{eq:lbint2}
&=  \int_{x\le 1}\exp(\delta)\exp\left(2 \pi \lambda c_2  \left(\frac{1-x^2}{2}\right)  \right)f_{d_0}(x) dx,
\end{align}
where $c_2 = \frac{\eta\exp\left(\frac{-1}{(M/\tau)} \right)}
{\left( 1-\eta\exp\left(\frac{-1}{(M/\tau)} \right)\right)}$. 

Evaluating the expectation with respect to the nearest BS distance $d_0$, whose PDF is given by Proposition \ref{prop:nnpdf}, we get $\bbE\{D\} $
\begin{eqnarray}\nn 
&\ge & \exp(\delta)\exp\left( \pi \lambda c_2\right) \int_{x\le 1} \exp \left(-\pi \lambda c_2 x^2\right)  f_{d_0}(x) dx, \\\nn
&=&\exp(\delta)\exp\left( \pi \lambda c_2\right)\\ \nn
&& \int_{x\le 1} \exp \left(-\pi \lambda c_2 x^2\right)  2 \pi \lambda x \exp(-\lambda \pi x^2) dx, \\\nn
&=&  \exp(\delta)\exp\left( \pi \lambda c_2\right) \\\label{eq:delaylbfinal}
&&\frac{1}{1+c_2}\left(1-\exp\left(-\pi \lambda (c_2+1)\right)\right).
\end{eqnarray}
Thus, \eqref{eq:ubound1} together with \eqref{eq:delaylbfinal} give the required lower bound on the expected delay stated in Theorem \ref{thm:lb}.

\begin{proposition}\label{prop:PGF} Let $\cG$ be the family of all non-negative, bounded measurable functions $g: \bbR^d \rightarrow \bbR$ on $\bbR^d$ whose support $\{x \in \bbR^d: g(x) > 0\}$ is bounded. Let $\cF$ be the family of all functions $f = 1-g,$  for $g\in\cG, \ 0\le g \le 1$. Then the \index{probability generating functional} probability generating functional for a point process $\Phi = \{x_n\}$ is defined as 
$$PGF(f) = \bbE\left\{\prod_{x_n \in \Phi} f(x_n)\right\}.$$
 For a homogenous PPP $\Phi$ with density $\lambda$, the probability generating functional is given by 
$$ PGF(f) = \exp^{-\int (1-f(x)) \lambda dx}.$$
\end{proposition}

\section{Proof of Corollary \ref{cor:bscord}}\label{app:cor:bscord}

The proof corresponding to the low-BS density regime \eqref{eq:newlb1c} that is obtained by ignoring the interference remains the same as the bound \eqref{eq:newlb1}, since at best if $k$ BSs are coordinating, the interference from the $k-1$ BSs can be removed at the $\sfm$.

Even when the interference is considered to get \eqref{eq:newlb2c}, most of the proof is identical to Theorem \ref{thm:lb}, where once again we make path-loss function $\ell(d_0)=1$ for the typical BS-MU link, and enforce restriction $1$. Supposing $k$ BSs are scheduling their transmissions together, then at best, the interference from the nearest $k-1$ BSs (other than $\sfn(\sfm)$ can be eliminated from $\SINR_{\sfn(\sfm), \sfm}(t)$ \eqref{eq:SINR}. Hence, following identical analysis as in Theorem \ref{thm:lb}, starting from \eqref{eq:lbint1}, where replacing the sum in the interference term from $z \in \Phi \backslash \{\sfn(\sfm)\}$ by $z \in \Phi \backslash \{k \ \text{nearest BSs from the origin}\}$,
till \eqref{eq:lbint2}, where replacing the 
nearest BS distance $d_0$ with the $k^{th}$ nearest BS distance $d_k$, we get 
$$
\bbE\{D\} \ge  \int_{x\le 1}\exp(\delta)\exp\left(2 \pi \lambda c_2  \left(\frac{1-x^2}{2}\right)  \right)f_{d_k}(x) dx.
$$
Taking the expectation with respect to the $k^{th}$ nearest BS distance $d_k$, whose PDF is given by Proposition \ref{prop:nnkpdf}, we get the same expected delay (upto constants) lower bound as in \eqref{eq:delaylbfinal}. Essentially, for deriving \eqref{eq:delaylbfinal}, we have shown that enough interferers lie in a unit disc around $\sfm$. Similar conclusion can be made even when considering interfering BSs other than the $k-1$ nearest ones as $\lambda$ increases. 

\begin{proposition}\cite{haenggi2005distances} \label{prop:nnkpdf} The cumulative distribution function and probability distribution function of the $k^{th}$ nearest BS distance $d_k$ is 
$$\bbP(d_k > y) =\sum_{i=1}^{k-1}\frac{(\lambda \pi y)^i}{i!}\exp(-\lambda \pi y^2),$$ and $$f_{d_k}(y) = 
\frac{2(\lambda \pi)^k y^{2k-1}}{(k-1)!} \exp(-\lambda \pi y^2).$$
\end{proposition}
\begin{proof} For $d_k > y$, there can be at most $k-1$ BSs in the disc of radius $y$ with node $\sfm$ as center located at the origin. Hence $$\bbP(d_k > y) = \sum_{i=1}^{k-1}\frac{(\lambda \pi y)^i}{i!}\exp(-\lambda \pi y^2),$$ and differentiating $\bbP(d_k > y)$, we 
get the mentioned PDF. 
\end{proof}

\section{Upper Bound on Capacity under CSIR}\label{app:ub}
As before, consider the restriction $1$ for all the non-typical BSs and the MUs served by them. Moreover, $\ell(d_0)=1$ for the typical MU.
Let $\Phi, \Phi_R$ be the sigma field generated by the BS point process $\Phi$, and the MU point process $\Phi_R$. Thus, we have the tail probability $\bbP(D > n | \Phi, \Phi_R) $
\begin{eqnarray}\nn
& = & \bbP\left(\tilde {\SINR}_\sfm(1) < \beta, \dots, \tilde {\SINR}_\sfm(n) < \beta \right | \Phi, \Phi_R), \\\label{app:eq:lbint200}
 &=& \prod_{t=1}^n\bbP\left(\tilde {\SINR}_\sfm(t) < \beta \bigl | \Phi, \Phi_R\right),
\end{eqnarray}
where the second equality follows since given $\Phi, \Phi_R$, the only randomness left in the SINRs ($\tilde {\SINR}$) is because of the fading gains 
$h$'s that are independent across time slots and all BS-MU pairs. 
Thus,
$\bbP\left(\tilde {\SINR}_\sfm(t) < \beta | \Phi, \Phi_R\right)$
\begin{equation} = \bbP\left(\frac{P_{\sfn(\sfm)}(t) h_{\sfn(\sfm)}(t) \1_{\sfn(\sfm)}(t)}{ \ga \sum_{z \in \Phi \backslash \{\sfn(\sfm)\}} \tilde {P}_z(t) \tilde{\1}_z(t)
h_z(t) \ell(z)} < \beta \bigl | \Phi, \Phi_R\right). \label{app:eq:lbint1}
\end{equation}

Taking the expectation with respect to $h_{\sfn(\sfm)}(t) \sim EXP(1)$ (note that $\1_{\sfn(\sfm)}(t)$ and $P_{\sfn(\sfm)}(t)$ cannot depend on $h_{\sfn(\sfm)}(t)$ under CSIR, which makes the analysis easy as follows), we get 
$\bbP\left(\tilde{\SINR}(t) < \beta | \Phi, \Phi_R\right) = 1-$
\begin{equation}\label{app:eq:lbint300}  \bbE\left\{\exp\left(\frac{-\beta \ga \sum_{z \in \Phi \backslash \{\sfn(\sfm)\}} \tilde{P}_z(t) \tilde{\1}_z(t)
h_z(t) \ell(z)  }{ P_{\sfn(\sfm)}(t) \1_{\sfn(\sfm)}(t)} \right) | \Phi, \Phi_R\right\}.
\end{equation}
Taking the expectation with respect to $\1_{\sfn(\sfm)}(t)$, 
$\bbP\left(\tilde{\SINR}(t) < \beta | \Phi, \Phi_R\right)$
\begin{eqnarray*} &=& 1- \bbE\left\{p_{\sfn(\sfm)}(t)\right.\\
&&\left. \exp\left(\frac{-\beta \ga \sum_{z \in \Phi \backslash \{\sfn(\sfm)\}} \tilde{P}_z(t) \tilde{\1}_z(t)
h_z(t) \ell(z)  }{ P_{\sfn(\sfm)}(t) }\right)| \Phi, \Phi_R\right\}.
\end{eqnarray*}

It is easy to check that for any $r \ge 1$, $r\exp(-c) \ge \exp(-c/r)$. 
Applying this identity to power $P_{\sfn(\sfm)}$ that is assumed to be $\ge 1$, we get 
$\bbP\left(\tilde{\SINR}(t) < \beta | \Phi, \Phi_R\right)$
\begin{eqnarray*} &\ge & 1- \bbE\left\{p_{\sfn(\sfm)}(t) P_{\sfn(\sfm)}(t)\right. \\
&&\left.\exp\left(-\beta \ga \sum_{z \in \Phi \backslash \{\sfn(\sfm)\}} \tilde{P}_z(t) \tilde{\1}_z(t)
h_z(t) \ell(z) \right) | \Phi, \Phi_R\right\}, \\
&=& 1- M \\ 
&&\bbE\left\{\exp\left(-\beta \ga \sum_{z \in \Phi \backslash \{\sfn(\sfm)\}} \tilde{P}_z(t) \tilde{\1}_z(t)
h_z(t) \ell(z)  \right)| \Phi, \Phi_R\right\},
\end{eqnarray*}
where the second equality follows since $p_{\sfn(\sfm)}(t) P_{\sfn(\sfm)}(t) \le M$.

Taking the expectation with respect to $h_z(t)$ and $\tilde{\1}_z(t)$, 
$\bbP\left(\tilde{\SINR}(t) < \beta | \Phi, \Phi_R\right) \ge 1- $
\begin{equation*} 
 M \bbE\left\{\prod_{z \in \Phi \backslash \{\sfn(\sfm)\}} \left(\frac{\tilde{p}_z(t)}{1+\tilde{P}_z(t) \ga \beta \ell(z)}   + (1-\tilde{p}_z(t))
   \right)| \Phi, \Phi_R\right\}.
\end{equation*}

Rearranging terms, $\bbP\left(\tilde{\SINR}(t) < \beta | \Phi, \Phi_R\right)$
\begin{align} \nn
\ge & 1- M \bbE\left\{\prod_{z \in \Phi \backslash \{\sfn(\sfm)\}} \left(1- \frac{\tilde{p}_z(t)\tilde{P}_z(t)\ga \beta\ell(z)}{1+\tilde{P}_z(t) \ga \beta \ell(z)}  
   \right)| \Phi, \Phi_R\right\}, \\ \nn
\ge & 1- M \bbE\left\{\prod_{z \in \Phi \backslash \{\sfn(\sfm)\}} \left(1- \frac{(M/\tau)\ga \beta\ell(z)}{1+ \frac{M}{\tilde{p}_z(t)} \ga \beta \ell(z)}  
  \right)| \Phi, \Phi_R\right\}, \\ \nn
\ge & 1- M \bbE\left\{\prod_{z \in \Phi \backslash \{\sfn(\sfm)\}} \left(1- \frac{\tilde{p}_z(t)(M/\tau)\ga \beta\ell(z)}{\tilde{p}_z(t)+ M \ga \beta \ell(z)}  
   \right)| \Phi, \Phi_R\right\}, \\\label{app:eq:dummyx1}
   \ge & 1- M \prod_{z \in \Phi \backslash \{\sfn(\sfm)\}} \left(1- \frac{\bbE\left\{\tilde{p}_z(t)| \Phi, \Phi_R\right\}(M/\tau)\ga \beta\ell(z)}{1+ M \ga \beta \ell(z)}  
   \right), 
\end{align}
where the second and the third inequalities follow since $\tilde{p}_z(t)\tilde{P}_z(t) \ge M/\tau$, and  $\tilde{p}_z(t)\tilde{P}_z(t) \le M$, respectively. For the final inequality, we bound the $\tilde{p}_z(t)\le 1$ in the denominator by $\tilde{p}_z(t)= 1$, note that $\tilde{p}_z(t)$ is independent for different BSs $z$. 
Denoting $\eta = \bbE\{\tilde{p}_z(t)\}$, and substituting \eqref{app:eq:dummyx1} into \eqref{app:eq:lbint200}, $\bbP(D > n | \Phi, \Phi_R)$
\begin{equation}\label{app:eq:eq3y}
 \ge  \prod_{t=1}^n\left(1- M \prod_{z \in \Phi \backslash \{\sfn(\sfm)\}} \left(1- \frac{\eta (M/\tau)\ga \beta\ell(z)}{1+ M \ga \beta \ell(z)}  
   \right)\right).
\end{equation}

Let $g= M\prod_{z \in \Phi \backslash \{\sfn(\sfm)\}}\left(1- \frac{\eta (M/\tau)\ga \beta\ell(z)}{1+ M \ga \beta \ell(z)}\right)$, where $g \le 1$. Then $\bbE\{D | \Phi\}$
\begin{eqnarray*} 
&= & \sum_{n=0}^\infty \bbP(D > n | \Phi, \Phi_R), \\
&\ge & \sum_{n=0}^\infty (1-g)^n,\ \text{from}\ \eqref{app:eq:eq3y}\\
&=& \frac{1}{g}, \\
&\ge& \frac{1}{M}\prod_{z \in \Phi \backslash \{\sfn(\sfm)\}} \left(\frac{1+ M \ga \beta \ell(z)}{1+M \ga \beta \ell(z)-\eta (M/\tau)\ga \beta \ell(z)}\right).
\end{eqnarray*}

To find $\bbE\{D\}$ from $\bbE\{D | \Phi\}$, we first use the probability generating functional (Proposition \ref{prop:PGF}) for the PPP $\Phi \backslash \{\sfn(\sfm)\}$ to get $\bbE\{D | d_0\}$
\begin{eqnarray*} 
&\ge & \frac{1}{M}\exp\left( \lambda  \int_{\bbR^2 \backslash \bB(0, d_0)} \left(\frac{1+ M \ga \beta \ell(z)}{1+M \ga \beta \ell(z)-\eta M\ga \beta\ell(z)} \right. \right.\\ 
&& \left. \left.
-1 \right) dz\right),\\
&=& \frac{1}{M} \exp\left( 2 \pi \lambda  \int_{\bbR, z>d_0} \left(\frac{1+ M \ga \beta \ell(z)}{1+M \ga \beta \ell(z)-\eta M\ga \beta\ell(z)} \right. \right.\\
&& \left. \left.
 -1 \right) z \ dz\right).
\end{eqnarray*}
Taking the expectation with respect to the nearest BS distance $d_0$, whose PDF is given by Proposition \ref{prop:nnpdf}, we get $\bbE\{D\} $
\begin{eqnarray*} 
&=& \int _x \frac{1}{M} \exp\left( 2 \pi \lambda  \int_{\bbR, z>d_0} \left(\frac{1+ M \ga \beta \ell(z)}{1+M \ga \beta \ell(z)-\eta M\ga \beta\ell(z)} \right. \right.\\
&& \left. \left.
 -1 \right) z \ dz\right)
f_{d_0}(x) dx
\end{eqnarray*}

Splitting the integral into two parts for $x \le 1$ and $x > 1$, and keeping only the part with $x\le 1$, where for $z < 1$, $\ell(z) = 1$, we get $\bbE\{D\}$
\begin{eqnarray*} 
&\ge & \int _{x\le 1} \frac{1}{M} \\
&&\exp\left(2 \pi \lambda \int_{x}^1 \left(\frac{1+ M \ga \beta }{1+M \ga \beta -\eta (M/\tau)\ga \beta} -1 \right) z \ dz   \right) \\ 
&&f_{d_0}(x) dx, \\
&= & \int _{x\le 1}\frac{1}{M} \\
&&\exp\left(2 \pi \lambda \left(\frac{1+ M \ga \beta }{1+M \ga \beta -\eta (M/\tau)\ga \beta} -1 \right)  \int_{x}^1 z \ dz   \right)\\
&& f_{d_0}(x) dx, \\
&= & \int _{x\le1}\frac{1}{M} \\ 
&&\exp\left(2 \pi \lambda \left(\frac{1+ M \ga \beta }{1+M \ga \beta -\eta (M/\tau)\ga \beta} -1 \right)  \left(\frac{1-x^2}{2}\right)  \right)\\
&&
f_{d_0}(x) dx. \\
\end{eqnarray*}
Let $c_8 = \left(\frac{1+ M \ga \beta}{1+M \ga \beta -\eta (M/\tau)\ga \beta} -1 \right)$. Then evaluating the expectation with respect to the nearest BS distance $d_0$, whose PDF is given by Proposition \ref{prop:nnpdf}, we get $\bbE\{D\} $
\begin{eqnarray}\nn 
&\ge& \frac{\exp\left( \pi \lambda c_8\right)}{M} \int_{x\le 1} \exp \left(-\pi \lambda c_8 x^2\right)  2 \pi \lambda x \exp(-\lambda \pi x^2) dx, \\\label{app:eq:delaylbfinal}
&=&  \frac{\exp\left( \pi \lambda c_8\right)}{M} \frac{1}{1+c_8}\left(1-\exp\left(-\pi \lambda (c_8+1)\right)\right).
\end{eqnarray}

\section{Proof of Theorem \ref{thm:ub}}\label{app:ach}
Without loss of generality, we will derive the upper bound on the expected delay for the typical user $\sfm$ that is served by its nearest BS $\sfn(\sfm)$ at a distance of $d_0$ from it.  We consider $k$ successive slots (not necessarily consecutive) that are dedicated for transmission to $\sfm$ by BS $\sfn(\sfm)$, and are interested in probability $\bbP(D > k)$ to upper bound the expected delay, where delay $D$ is as defined in Definition \ref{defn:delay}.

Typically, the number of MUs connected to different BSs are different. Thus, during the $k$ considered slots for $\sfm$, any BS other than $\sfn(\sfm)$ transmits potentially to different MUs. 
Let $\cG_k$ be the sigma field generated by the BS point process $\Phi$ and MU point process $\Phi_R$
and the choice (index) of MUs being served by BSs of $\Phi$ at the above described $k$ slots 
$t=1,2, \ldots ,k$. 

With the above mentioned local power control strategy \eqref{eq:powercontrol}, the SINR seen at $\sfm$ in time slot $t$ is 
\begin{equation}\label{eq:newSINR}
\SINR_{n(\sfm), \sfm}(t) = \frac{ c h_{\sfn(\sfm)}(t)  \1_{\sfn(\sfm)}(t)}{\ga I(t)  + \sfN},
\end{equation}
where
\begin{equation}\label{eq:redintf}
I(t) = \sum_{z \in \Phi \backslash \{ \sfn(\sfm) \}} \1_z(t)
P_z(t) h_z(t) \ell(z).
\end{equation}
With 
$e_{\sfn(\sfm), \sfm}(t) =1$ if $\SINR_{(\sfm,\sfn(\sfm))}(t) > \beta$, and $0$
otherwise, we have 

\begin{equation}\label{eq:dummy1}
 \bbP\left[D > k \big| \cG_k \right] =  \bbE \left\{\bbP\left[e_{\sfn(\sfm), \sfm}(t) =0, \; \forall \;
t=1,\dots, k \big| \cG_k \right]  \right\}.
\end{equation}

Given $\cG_k$, with the described strategy \eqref{eq:powercontrol}, the
transmission events $\1_{z}(t)$, and the transmit powers $P_z(t)$ are independent across time slots $t$ for all 
BSs $z$. Moreover, the fading 
gains $h_{(.)}(t)$ are all independent. Hence, we get 
\begin{equation}
\label{eq:indep} \bbP\left[D > k \big| \cG_k \right]  =  \bbE \left\{
\prod_{t=1}^k \bbP \left[ e_{\sfn(\sfm), \sfm}(t) =0 \big| \cG_k \right] \right\}.
\end{equation}
Let $A(t)$ be the event that $ \1_{\sfn(\sfm)}(t) =0$, i.e., the BS $\sfn(\sfm)$ does not transmit in the slot designated for user $\sfm$, while $B(t)$ be the event that 
$\{ \1_{\sfn(\sfm)}(t)=1\} \cap e_{\sfn(\sfm), \sfm}(t) =0$, i.e., BS $\sfn(\sfm)$ transmits in slot $t$ but the transmission fails. Then 
\begin{equation}
\label{eq:indep2} \bbP\left[D > k \big| \cG_k \right]  =  \bbE \left\{
\prod_{t=1}^k \bbP \left[ A(t) \cup B(t) \big| \cG_k \right] \right\}.
\end{equation}
Because of our power transmission strategy, given $\cG_k$, transmission event $\1_{\sfn(\sfm)}(t)$ and the success event $e_{\sfn(\sfm), \sfm}(t)$ are independent, hence
$\bbP(A(t)|\cG_k) = 1 - p_{\sfn(\sfm)}(t)$, while $\bbP(B(t)|\cG_k)$
\begin{equation} \label{eq:E2}
 = p_{\sfn(\sfm)}(t)
\left(1-\bbE\left\{\exp\left(- \f{\beta}{c} (\sfN+ \ga
I(t))\right)\bigg{|}\cG_k\right\}\right),
\end{equation}
that follows by taking expectation with respect to $h_{\sfn(\sfm)}(t) \sim
EXP(1)$. 
Using the union bound, from \eqref{eq:indep2}, we get $\bbP \left[ A(t) \cup B(t) \big| \cG_k \right]$
\begin{eqnarray} \nn
 & \leq & 1- p_{\sfn(\sfm)}(t) + p_{\sfn(\sfm)}(t)\\\nn
 &&
 \left(1-\bbE\left\{\exp\left(- \f{ \beta}{c} (\sfN+ \ga I(t))\right)\Big{|}\cG_k \right\}\right),\\\nn
& \leq & 1 - p_{\sfn(\sfm)}(t) \exp\left(- \f{ \beta \sfN}{c}
\right)\, \\ \label{eq:E3} 
&&\bbE\left\{ \exp\left(- \f{\beta \ga}{c}I(t)\right) |\cG_k \right\}.
\end{eqnarray}

Let $a =
\f{ \beta \ga}{c}$, and we focus on finding a lower bound on $\bbE\left\{ \exp\left(- 
a I(t)\right) |\cG_k \right\}$, that is independent of
the choice of the MU being served by BS $z$. To this end,  
we first expand $\bbE\left\{ \exp\left(- aI(t)\right) \Big|\cG_k \right\}$
\begin{equation} \label{eq:mgfint}
 = \prod_{z \in \Phi \backslash \{ \sfn(\sfm) \}} \bbE\left\{ \exp\left(- a \1_z(t)
p_z^{(u(z))}(t) h_z(t) \ell(z)\right) \Big|\cG_k \right\}
\end{equation}
where BS $z \in \Phi \backslash \{\sfn(\sfm) \}$ transmits to the MU $u(z)$
 in time slot $t$. This fixes the transmission
probability $p_z^{(u(z))}(t)$ and power $P_z^{(u(z))}(t)$ (where we have included the index $u$ to make the dependence on the MU explicit). Then, taking the expectation with respect to $ \1_z(t)$, we have
\begin{align} \nn \lefteqn{ \bbE\left\{\exp\left(-a \1_z(t) P_z^{(u(z))}(t)
h_z(t) \ell(z)\right) \Big| \cG_k \right\}} \\
\nn  = & (1-p_z^{(u(z))}(t))  \\\nn
&+  p_z^{(u(z))}(t)\bbE\left\{ \exp\left( -
a P_z^{(u(z))}(t) h_z(t) \ell(z)\right) \Big| \cG_k \right\} \\ \label{eq:bound_indiv_int} 
 = & (1-p_z^{(u(z))}(t)) + p_z^{(u(z))}(t) \f{1}{1
+ a  \ga \ell(z) P_z^{(u(z))}(t)},
\end{align}
where the second equality follows by taking expectation with respect to the independent fading gains $h_z(t) \sim EXP(1)$.

Let $u^*(z)$ be the MU for which the right hand
expression in \eqref{eq:bound_indiv_int} is minimized, i.e., BS $z$ causes maximum interference at $\sfm$ when it is serving MU $u^*(z)$. Let
$p_z^*$, $P_z^*$ denote the corresponding transmission probability
and power, respectively for BS $z$. Denote by $\1^*_z$ an independent Bernoulli random
variable with $\bbP[\1^*_z = 1] =p_z^*$.  Define
\begin{equation}
I^*(t) = \sum_{z \in \Phi \backslash \{ \sfn(\sfm) \}}
\1^*_z P_z^* h_z(t) \ell(z). \label{eq:interference_max}
\end{equation}
Essentially $I^*(t)$ dominates the actual interference $I(t)$ seen at $\sfm$. 
Substituting $I^*(t)$ for $I$ in \eqref{eq:E3} along with the observation that
given $\Phi, \Phi_R$, $I^*(t) \stackrel{d}{=} I^*(1)$, and we get 

$\bbP \left[ A(t) \cup B(t) \big|\cG_k \right] $
\begin{eqnarray}\nn
&\leq& 1 - p_{\sfn(\sfm)}(t)\\
 &&\exp\left(- \f{ \beta
\sfN}{c}\right) \, \bbE\left\{ \exp\left(- a I^*(1)\right) \Big| \Phi, \Phi_R \right\}.
\label{eq:E4}
\end{eqnarray}
Let
\[ \theta = \exp\left(- \f{ \beta
\sfN}{c}\right) p_{\sfn(\sfm)}(t)\, \bbE\left\{ \exp\left(- a I^*(1)\right) \Big| \Phi, \Phi_R \right\}.\]
Substituting from \eqref{eq:E4} in \eqref{eq:indep}, we get
\begin{equation}\label{eq:J}
\bbP\left[D >  k \big| \Phi, \ \Phi_R\right]  \leq  (1-\theta)^k.
\end{equation}

Then the expected delay can then be written as
\begin{eqnarray*}\nn
\bbE\{D\}& = & \sum_{k\ge 0} \bbP[D >  k] \\
\nn & = & \bbE\left\{\sum_{k\ge 0} \bbP\left[D > k \big| \Phi, \Phi_R \right]\right\} \\
\nn & \leq & \bbE\{\theta^{-1}\},
\end{eqnarray*}
where the last inequality follows from \eqref{eq:J}.
Using the Cauchy-Schwartz inequality, on the two random variables in $\theta^{-1}$, we get
$ $
\begin{eqnarray} \nn
\bbE\{D\}  &\le&  \exp\left(\beta \sfN/c\right) \left(\bbE\left\{\f{1}{ \left( \bbE\left\{ \exp\left(-a
I^*(1)\right) \big| \Phi, \Phi_R\right\}\right)^2 } \right\} \right. \\\label{eq:CS}
&& \left.\bbE\left\{p_{\sfn(\sfm)}(t)^{-2}\right\} \right)^{\f{1}{2}}.
\end{eqnarray}
From the definition of the transmission probability $p_{\sfn(\sfm)}(t) = \left(M/c\right) \ell(d_0) = \left(M/c\right) \min\{1,d_0^{-\al}\}$, we
get $\bbE[p_{\sfn(\sfm)}(t)^{-2}] $
\begin{align}\nn
&= \left(\f{c}{M} \right)^{2}\left(\int_0^1 1 f_{d_0}(x) dx + \int_1^\infty d_0^{2 \al}f_{d_0}(x) dx\right), \\\nn
& \le \left(\f{c}{M} \right)^{2}\left(1 + \int_0^\infty d_0^{2 \al}f_{d_0}(x) dx\right), \\ \label{eq:p(o)bound}
&=\left(\f{c}{M} \right)^{2} \left(1+ \frac{\Gamma(\alpha+1)}{(\pi \lambda)^\alpha}\right),
\end{align}
using the PDF of $d_0$ from Proposition \ref{prop:nnpdf}.

Now we work towards bounding $\bbE\left\{\f{1}{ \left( \bbE\left\{ \exp\left(-a I^*(1)\right) \big| \Phi, \Phi_R \right\} \right)^2
} \right\}$.
Recall, $\bbE\left\{ \exp\left(-a I^*(1)\right) \big| \Phi, \Phi_R \right\}$
\begin{equation}
  = \prod_{z \in \Phi
\backslash \{\sfn(\sfm) \}} \bbE \left\{ \exp\left(- a \1^*_z P_z^*
h_z(1) \ell(z)\right) \big| \Phi, \Phi_R  \right\}.
\label{eqn:cond_exp_int}
\end{equation}
Taking expectations, first with respect to $\1^*_z$ and then with respect to fading gain $h_z(1)$, we get $\bbE \left\{ \exp\left(- a \1^*_z P_z^* h_z(1) \ell(z)\right) | \Phi, \Phi_R
\right\} $
\begin{eqnarray} \nn
& = & (1-p^*_z) + p^*_z \bbE \left\{ e^{- a
P_z^* h_z(1) \ell(z)} \big| \Phi, \Phi_R \right\} \\ \nn
 & = & 1 - p_z^* \left( 1 - \f{1}{1 + a P_z^* \ell(z)}
\right) \\ \nn
 & \stackrel{(a)}= & 1 - \f{\beta \ga p_z^* P_z^* \ell(z)}{c + \beta \ga P_z^*
\ell(z)} \\ \label{eq:dummyyy}
 & \stackrel{(b)}\geq & 1 - \f{\beta \ga M \ell(z)}{c} = 1 - \beta \ga (1-\ep)
 \ell(z),
\end{eqnarray}
where $(a)$ follows by resubstituting $a =\f{\beta \gamma}{c} $, and $(b)$ by invoking the average power constraint  $p_z P_z \le M$ for $\forall \ z$ and in 
particular $p^*_z P^*_z \le M$ and $c = M(1-\ep)^{-1}$. Let $\kappa=\beta \ga (1-\ep)$, where note that because of assumption on 
$\ep$ ($\ep$ satisfies $(1-\epsilon)\beta \gamma < 1$)  in the power control strategy \eqref{eq:powercontrol}, $\kappa <1$. Substituting \eqref{eq:dummyyy} in (\ref{eqn:cond_exp_int}), we get
\begin{equation}
\bbE\left\{ e^{-a I^*(1)} \big| \Phi, \Phi_R \right\} \geq \prod_{z \in \Phi
\backslash \{\sfn(\sfm) \}} \left( 1 - \kappa \ell(z) \right).
\label{eqn:bound_cond_exp_int}
\end{equation}
%
Hence $\bbE\left\{\f{1}{ \left( \bbE\left\{\exp\left(-a I^*(1)\right)| \Phi,\Phi_R \right\}
\right)^2 } \right\}$\begin{equation}
  \le  \bbE\left\{ \prod_{z \in \Phi
\backslash \{\sfn(\sfm) \}) } \exp\left( - 2 \log \left( 1 - \kappa
\ell(z) \right)\right) \right\}. \label{eq:expintfbound1}
\end{equation}
Once again using the probability generating functional (Proposition \ref{prop:PGF}) for the PPP $
 \Phi
\backslash \{\sfn(\sfm) \}$, we get $\bbE\left\{\f{1}{ \left( \bbE\left\{\exp\left(-a I^*(1)\right)| \Phi \right\}
\right)^2 } \right\}$
\begin{eqnarray}\nn
  & \leq & \bbE_{d_0}\left\{\exp\left( \lambda \int_{\bbR^2 \backslash \bB(0,d_0)}
\left( \exp\left(-2 \log \left( 1 - \kappa \ell(z) \right)\right) \right.\right.\right.\\
&&\left.\left.\left.
- 1 \right) {\mathrm d}z \right)\right\}, \nn \\
 & \leq &\bbE_{d_0}\left\{ \exp \left( \f{2 \lambda \kappa}{(1-\kappa)^2} \int_{\bbR, z > d_0} z\ell(z)
{\mathrm d}z \right)\right\} , \label{eqn:finite_exp_int}
\end{eqnarray}
where $\bB(0,d_0)$ is the disc with radius $d_0$ centered at the origin, and the last inequality follows by noting that $\ell(z) \le 1$.  Similar to the proof of Theorem \ref{thm:lb}, we can separate the integral into two parts $z > d_0 <1$ (where $\ell(z) = 1$) and $z > d_0> 1$ (where $\ell(z) = z^{-\alpha}$), to get 
$\bbE\left\{\f{1}{ \left( \bbE\left\{\exp\left(-a I^*(1)\right)| \Phi \right\}
\right)^2 } \right\}$
\begin{eqnarray}\nn
 & \leq &\int_{0}^{1}  \exp \left( \f{2 \lambda \kappa}{(1-\kappa)^2} \left(\frac{1-x^2}{2}\right)\right) f_{d_0}(x) dx\\ \nn
 &&+
\int_{1}^\infty  \exp \left( \f{2 \lambda \kappa}{(1-\kappa)^2} \left( \frac{x^{2-\alpha}}{\alpha-2}\right)\right) f_{d_0}(x) dx, \label{eqn:finite_exp_int2}
\end{eqnarray}

Since $\alpha > 2$, using Proposition \ref{prop:nnpdf}, 
we get the following bound on the expectation 
\begin{eqnarray}
 \bbE\left\{\f{1}{ \left( \bbE\left\{\exp\left(-a I^*(1)\right)| \Phi \right\}
\right)^2 } \right\}& \leq &\exp \left(c_4\lambda\right),
\end{eqnarray}
where $c_3$ is a constant.
Combining this with \eqref{eq:p(o)bound}, from \eqref{eq:CS} we get 
$$\bbE\{D\} \le \sqrt{c_5 \left(1+ \frac{\Gamma(\alpha+1)}{(\pi \lambda)^\alpha}\right) \exp \left(c_4\lambda\right)},$$
where $c_4$ is a constant.
This completes the proof of Theorem~\ref{thm:ub}.

\end{document}